\documentclass[aps,pra,amsmath,twocolumn,amssymb,superscriptaddress]{revtex4-1}

\makeatletter
\def\@bibdataout@aps{%
 \immediate\write\@bibdataout{%
  @CONTROL{%
   apsrev41Control,author="08",editor="1",pages="0",title="0",year="1",eprint="1"%
  }%
 }%
 \if@filesw
  \immediate\write\@auxout{\string\citation{apsrev41Control}}%
 \fi
}%
\makeatother 

\usepackage{amssymb}
\usepackage{amsthm}
\usepackage{amsfonts}
\usepackage{complexity}
\usepackage{graphicx}
\usepackage{dcolumn}
\usepackage{bm}
\usepackage{hyperref}
\usepackage{enumerate}
\usepackage{algorithm}
\usepackage{algpseudocode}

    \newcommand{\linecomment}[1]{\State \(\triangleright\) {\footnotesize #1} \normalsize}

\usepackage{multirow}

\newtheorem{theorem}{Theorem}

\input{Qcircuit.tex}

\def\ket#1{\left|#1\right\rangle}

\newcommand{\CRej}{\text{RejF }}

\newcommand{\reset}{\mathrm{reset}}

\newcommand{\eq}[1]{\hyperref[eq:#1]{(\ref*{eq:#1})}}
\renewcommand{\sec}[1]{\hyperref[sec:#1]{Section~\ref*{sec:#1}}}
\newcommand{\app}[1]{\hyperref[app:#1]{Appendix~\ref*{app:#1}}}
\newcommand{\fig}[1]{\hyperref[fig:#1]{Figure~\ref*{fig:#1}}}
\newcommand{\thm}[1]{\hyperref[thm:#1]{Theorem~\ref*{thm:#1}}}
\newcommand{\lem}[1]{\hyperref[lem:#1]{Lemma~\ref*{lem:#1}}}
\newcommand{\tab}[1]{\hyperref[tab:#1]{Table~\ref*{tab:#1}}}
\newcommand{\cor}[1]{\hyperref[cor:#1]{Corollary~\ref*{cor:#1}}}
\newcommand{\alg}[1]{\hyperref[alg:#1]{Algorithm~\ref*{alg:#1}}}
\newcommand{\defn}[1]{\hyperref[def:#1]{Definition~\ref*{def:#1}}}

\newcommand{\tout}[1]{{}}

\newcommand{\etal}{\emph{et al.}}
\newcommand{\ii}{\mathrm{i}}
\newcommand{\ee}{\mathrm{e}}

\usepackage[usenames,dvipsnames]{xcolor}
\hypersetup{
    colorlinks=true,       
    linkcolor=Maroon,          
    citecolor=OliveGreen,        
    filecolor=magenta,      
    urlcolor=Blue           
}

\begin{document}


\title{Efficient Bayesian Phase Estimation}
\author{Nathan Wiebe}
\affiliation{Quantum Architectures and Computation Group, Microsoft Research, Redmond, WA (USA)}

\author{Chris Granade}
\affiliation{Centre for Engineered Quantum Systems, Sydney, NSW (Australia)}
\affiliation{School of Physics, University of Sydney, Sydney, NSW (Australia)}

\begin{abstract}
    We provide a new efficient adaptive algorithm for performing phase
    estimation that does not require that the user infer the bits of the
    eigenphase in reverse order; rather it directly infers the phase and
    estimates the uncertainty in the phase directly from experimental data. Our
    method is highly flexible, recovers from failures, and can be run in the
    presence of substantial decoherence and other experimental imperfections
    and is as fast or faster than existing algorithms.
\end{abstract}

\maketitle

\section{Introduction}
\label{sec:intro}

Eigenvalue estimation has been a cornerstone of physics since the dawn of spectroscopy.  In recent years, ideas from quantum information have revolutionized the ways that we estimate these values by providing methods that require exponentially fewer experiments than statistical sampling and total experimental time that saturates the Heisenberg limit, which is the best possible scaling allowed by quantum mechanics.  This quantum approach, known as phase estimation (PE), is crucial to acheive many of the celebrated speedups promised by quantum computing~\cite{shor_polynomial-time_1995,BHM+02,ADL+05,harrow2009quantum,lanyon2010towards} and as such optimizing these algorithms is especially important for present day experimental demonstrations of such algorithms because the number of quantum gates that can be performed is often severely limitted by decoherence.

The most popular approach to PE is known as iterative phase estimation (IPE),
as it forgoes the use of quantum resources to infer the eigenvalues of unitary
matrix in favor of using a classical inference algorithm to estimate the
eigenvalues from experimental
data~\cite{Kit96,kitaev2002classical,higgins2007entanglement,SHF14}.  Although
the methodologies used for inference stretch back to the nineties,
the ongoing revolution in machine learning that has lead to a plethora of
improved classical inference procedures.  This raises an important question:
can these ideas be used to speed up phase estimation or make it more robust to
experimental error?

We answer this in the affirmative by providing a new classical inference
method, inspired by recent work on particle filter methods, that is tailored
to phase estimation. Our method not only provides performance
advantages over existing iterative methods~\cite{Kit96,kitaev2002classical}, but
is also robust to experimental imperfections such as depolarizing noise and
small systematic errors.

Before going into detail about our algorithm we will first review the phase
estimation circuit and Bayesian approaches to PE.  We then show that rejection
sampling can be used to efficiently approximate Bayesian inference, thereby
allowing PE to inherit the speed and robustness of Bayesian approaches while
retaining the efficiency of traditional methods.

Iterative PE infers the eigenvalue of a given eigenvector of a unitary matrix
$U$ from a set of experiments that are performed on the circuit
\begin{equation*}
    \Qcircuit @C=1em @R=1em {
        \lstick{\ket{0}}    & \gate{H}  & \gate{Z(-M \theta)}   & \ctrl{1}   & \gate{H} & \meter & \cw \\
        \lstick{\ket{\phi}} & {/} \qw   & \qw                   & \gate{U^M} & \qw      & \qw    & \qw
    }
\end{equation*}
where $U\ket{\phi} = \ee^{\ii\phi}\ket{\phi}$ for an unknown eigenphase $\phi \in \mathbb{R}$,
and where $Z(M \theta) = \ee^{\ii M \theta Z}$ is a rotation that can be used to compare the eigenphase of $U$ to
a known reference value.
Most PE algorithms use this circuit to infer the bits
in a binary expansion of $\phi$ in order from least significant to most significant.  The process is can be made near optimal (up to $\log^*$ factors~\cite{SHF14}) and
affords an efficient inference method for these bits.


\section{Bayesian Phase Estimation}
\label{sec:bayesian-phase-est}

Bayesian phase estimation, as introduced by Svore \etal~\cite{SHF14}, involves performing a 
set of experiments and then updating the prior distribution using Bayes' rule.
For example, if an experiment is performed with using $M$ repetitions of $U$,
$Z(M \theta)$ and a measurement outcome $E\in \{0,1\}$ is observed then Bayes'
rule states that the \emph{posterior probability} distribution for $\phi$
after observing the datum is
\begin{equation}
P(\phi|E;\theta,M) = \frac{P(E|\phi;\theta,M)P(\phi)}{\int P(E|\phi;\theta,M)P(\phi)\mathrm{d}{\phi}}.\label{eq:update}
\end{equation}
The final ingredient that is needed to perform an update is the likelihood
function $P(0|\phi;\theta,M)$. For phase estimation,
\begin{gather}
    \label{eq:likenodecohere}
    \begin{aligned}
        P(0|\phi;\theta,M) & = \frac{1+\cos(M[\phi +\theta])}{2},\\
        P(1|\phi;\theta,M) & = \frac{1-\cos(M[\phi +\theta])}{2}.
    \end{aligned}
\end{gather}
After using~\eq{update} to update the posterior distribution we then set the prior distribution to equal the posterior distribution.  This process is then repeated for each of the random experiments in the data set.

Unlike conventional methods, Bayesian inference returns a posterior
distribution over the phase. The mean and standard deviation of this
distribution provide an estimate of the true eigenvalue and the algorithm's
uncertainty in that value. More sophisticated estimates of uncertainty, such
as credible regions, can also be extracted from the posterior
distribution~\cite{granade_robust_2012,ferrie_high_2014}. Estimates of
the uncertainty are crucial here because we use them to design highly informative experiments 
 and also because they allow the protocol to be terminated when a threshold accuracy has been reached.

\begin{figure}[t!]
    \begin{centering}
        \includegraphics[width=0.8\linewidth]{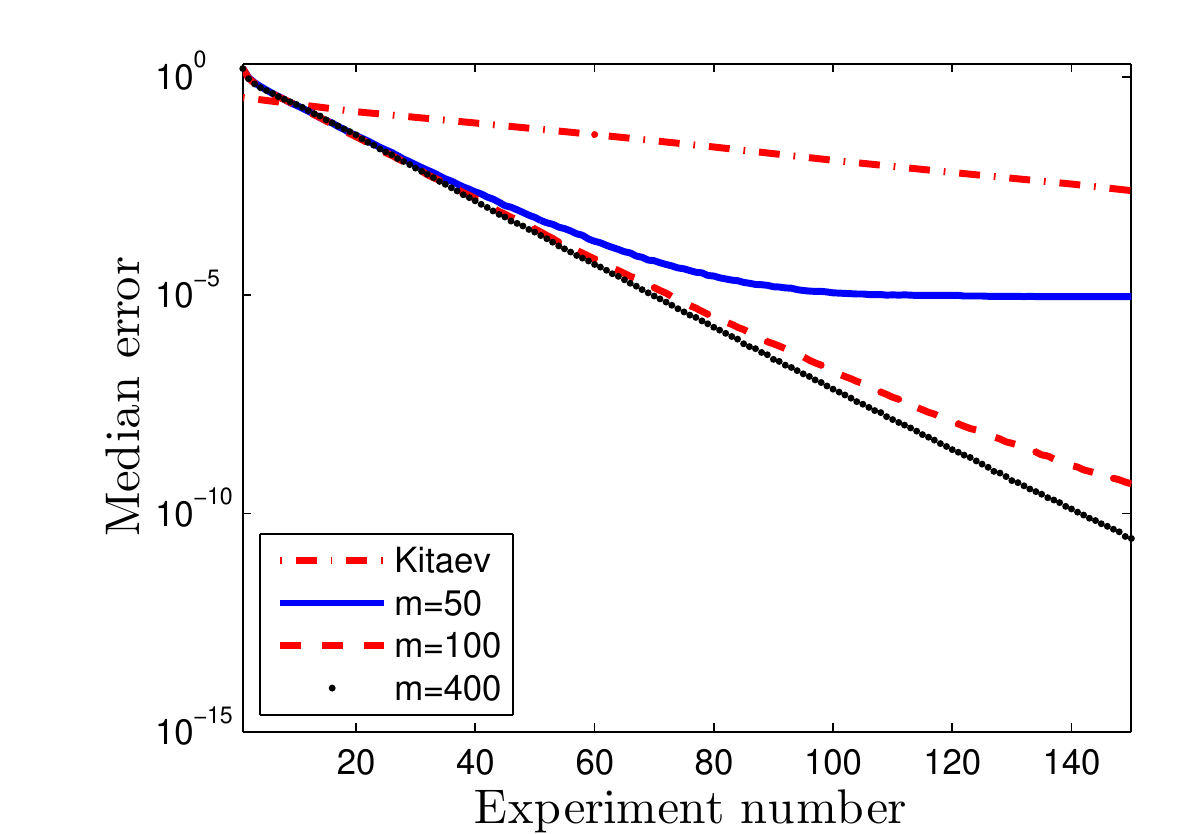}
    \end{centering}
    \caption{\label{fig:PEerror}
     Median errors in phase estimation for 10~000 random initial choices of the true eigenphase.
    }
\end{figure}

\section{Approximate Bayesian Phase Estimation}

Exact Bayesian inference is impossible if the eigenphases are
continuous so approximations are needed to
make the procedure tractable.  Rather than na\"ively discritizing the prior
distribution, modern methods discretize by sampling from the prior and
then perform Bayesian inference on the discrete set of samples (often called
``particles'').  These particle filter methods methods are quite powerful and
have become a mainstay in computer vision and machine
learning~\cite{haykin2004kalman,smith2013sequential,isard_condensationconditional_1998}.

Despite the power of these methods, they are often hard to implement and even
harder to deploy in a memory limited environment such as on an embedded
controller or an FPGA. With the increasing use of FPGAs in the control of
quantum information experiments
\cite{shulman_suppressing_2014,casagrande_design_2014,hornibrook_cryogenic_2015},
overcoming this limitation presents a
significant advantage to modern experiment design. 

We propose
 a much simpler approach that we call Rejection Filtering Phase
Estimation (RFPE). Rather than using a set of hypotheses that implicitly
define a model for the system, we posit a prior model and directly update it
to find a model for our posterior distribution.  We achieve this by using a
Gaussian with mean $\mu$ and variance $\sigma^2$ to model our initial prior, perform a Bayesian update on samples
drawn from the distribution and then refit the updated samples to a Gaussian.
This strategy is used in a number of particle filter methods such as the extended Kalman filter and assumed density
filtering~\cite{haykin2004kalman,opper1998bayesian}.  We further optimize
this process by approximating the Bayesian inference step using rejection
sampling.  This can reduce the memory required by a factor of $1~000$ or more,
as we need only consider a single sample at a time. Our algorithm is described
below and pseudocode is given in \app{pseudocode}.

\begin{figure}[t!]
    \begin{centering}
        \includegraphics[width=0.723\linewidth]{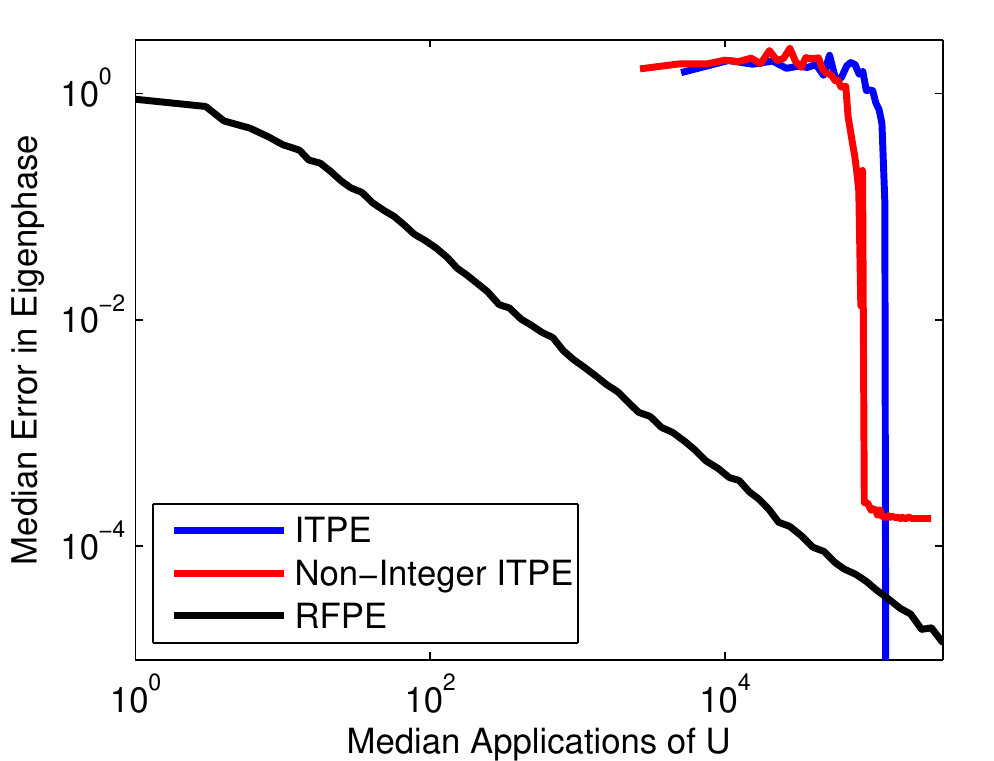}
    \end{centering}
    \caption{\label{fig:ITPEcmp}
     Comparison of RFPE to ITPE for $t=10~000$ with $100$ samples for $\phi_{\rm true} = 2\pi k/t$ at each measurement.  
    }
\end{figure}

\begin{figure*}
    \begin{centering}
\includegraphics[width=0.45\linewidth]{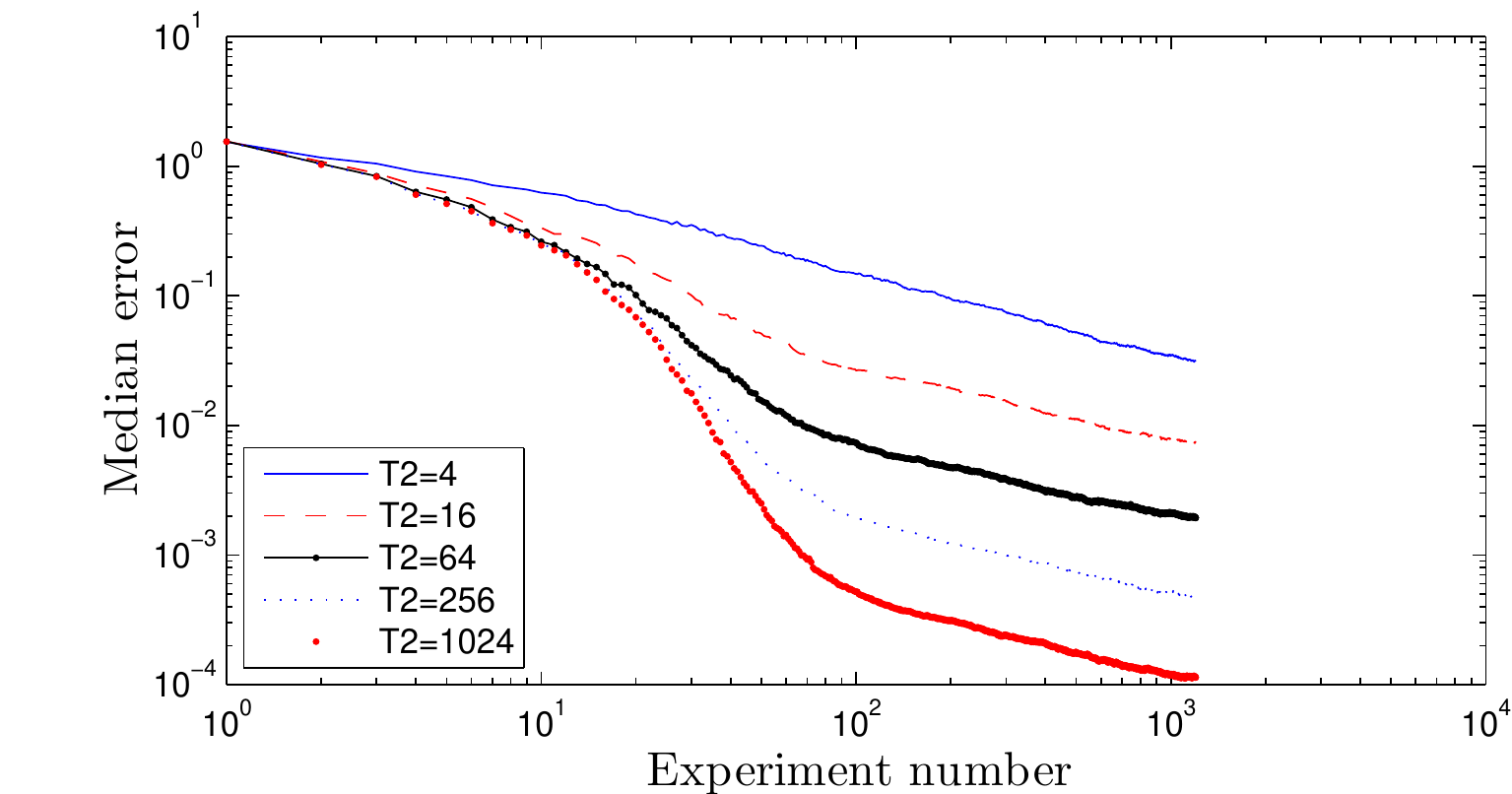}
        \includegraphics[width=0.45\linewidth]{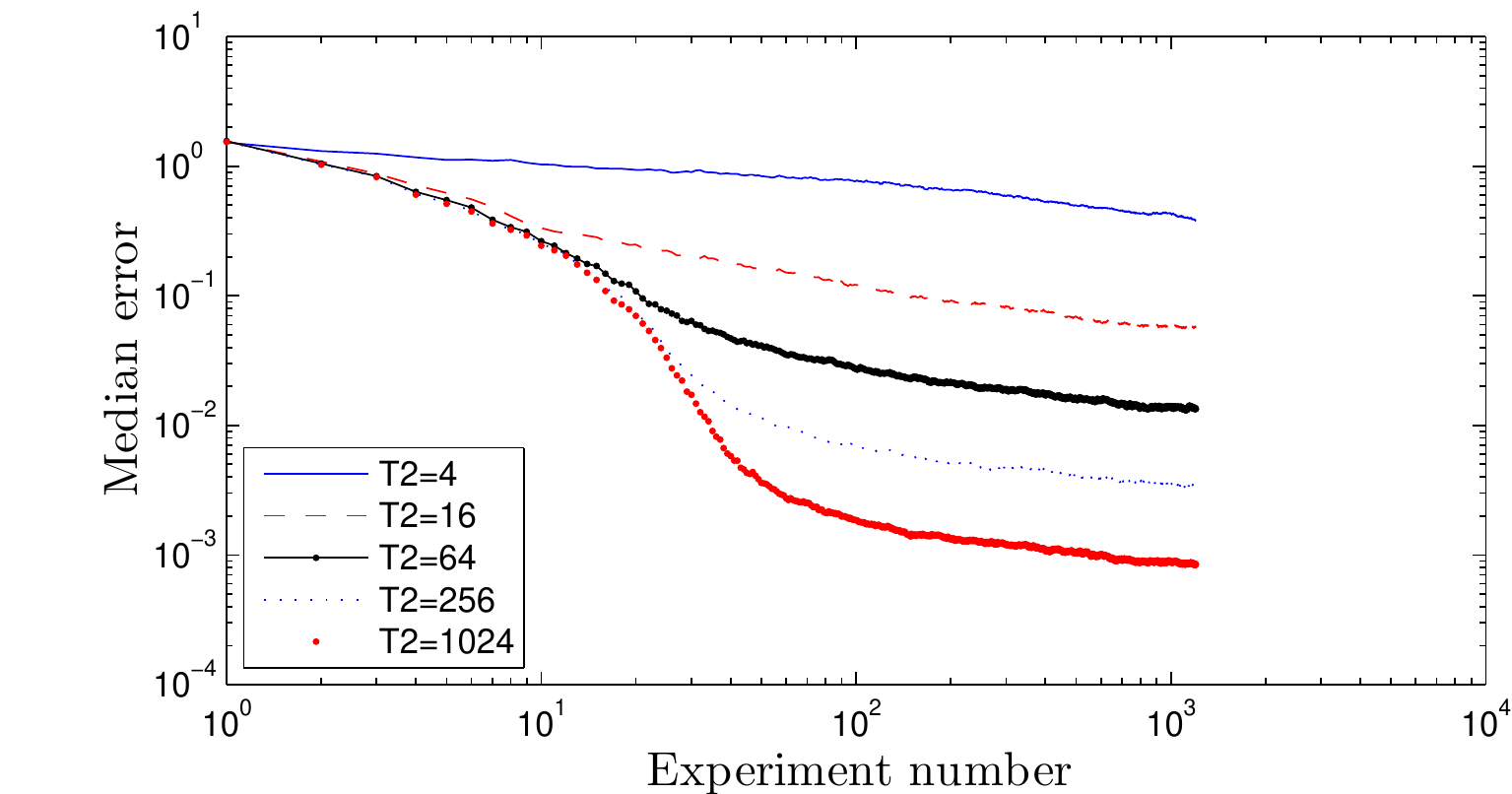}
    \end{centering}
    \caption{\label{fig:T2plot}
Median errors for phase estimation in decohering systems for experiments constrained to use $M\le T_2$ (left) and $M\le T_2/10$ (right).  We take $m=12~000$ and use $1~000$ random samples to generate the data above.  The initial state is taken to be a randomly chosen eigenstate in all cases.
    }
\end{figure*}

\begin{enumerate}
\item Perform experiment for given $\theta$, $M$ and observe outcome $E\in \{0,1\}$.
\item Draw $m$ samples from $\mathcal{N}(\phi|\mu,\sigma^2)$.
\item For each sample, $\phi_j$, assign $\phi_j$ to $\Phi_{\rm accept}$ with probability $P(E|\phi_j;\theta,M)/\kappa_E$, where $\kappa_E\in (0,1]$ is a constant s.t. $P(E|\phi_j;\theta,M)/\kappa_E\le 1$ for all $\phi_j,E$.
\item Return $\mu = \mathbb{E}(\Phi_{\rm accept})$ and $\sigma =\sqrt{\mathbb{V}(\Phi_{\rm accept})}$.
\end{enumerate}

The resultant samples are equivalent to those drawn from the posterior distribution
$P(\phi|E;M,\theta)$.  To see this, note that the probability density of a sample being accepted at $\phi=\phi_j$ is $ P(E | \phi; \theta, M) \mathcal{N}(\phi|\mu,\sigma^2)$.  Eqn~\eq{update} then implies 
\begin{equation}
    P(E | \phi; \theta, M) \mathcal{N}(\phi|\mu,\sigma^2) \propto P(\phi | E; \theta, M),
\end{equation}
which implies that the accepted samples are drawn from the posterior distribution.  

Although it is difficult to concretely predict the value of $m$ needed to make the error in the inference small, we show in \app{stability} that $m$ must scale at least as the inverse square of the relative fluctuations in the likelihood function.  Similarly, Markov's inequality shows that $m$ must scale at least as $\kappa_E/\int P(E|\phi;\theta,M)P(\phi) \mathrm{d}\phi$ to ensure that, with high probability, the mean can be accurately estimated from the accepted samples.  We introduce $\kappa_E$ to compensate for small expected likelihoods.

The main issue that remains is how to optimally choose the parameters $\theta$
and $M$. One approach is to locally optimize the Bayes
risk~\cite{granade_robust_2012}, but the resulting calculation can be too
expensive to carry out in online experiments that provide experimental results
at a rate of tens of megahertz or faster.  Fortunately, the particle guess heuristic (PGH) can give an
expedient and
near-optimal experiment for this class of likelihood
functions~\cite{wiebe_hamiltonian_2014},
\begin{align}
    M &= \left\lceil\frac{1.25}{\sigma}\right\rceil,~
    \theta \sim P(\phi).\label{eq:PGH}
\end{align}
The factor of $1.25$ comes from optimizing the cost of RFPE.   Non-integer $M$
are appropriate if $U=\ee^{-i H M}$.

\fig{PEerror} shows the error incurred using RFPE.  The most obvious feature is that the error shrinks exponentially with the number of experiments (which is proportional to the evolution time under the PGH) for $m>100$.  Roughly $150$ experiments are needed for the method to provide $32$ bits of accuracy in the median case.
We discuss the scaling in the mean in \app{var-reduction}.

The number of experiments needed to reach error $\epsilon$ scales as $O(\log(1/\epsilon))$ rather than $O(\log(1/\epsilon)\log\log(1/\epsilon))$ in Kitaev's method~\cite{Kit96,kitaev2002classical}.  Concretely, after $150$ experiments the median error for Kitaev's PE algorithm (with $s=10$) is roughly $10$ million times that of RFPE.

Although the number of experiments needed is small, the experimental time need
not be.  If the error shrinks as $\epsilon\in \Theta(\ee^{-\lambda N})$ where
$N$ is the experiment number then $T_{\exp}\in O(\sum_{N=1}^{N_{\max}} \ee^{\lambda N})\in O(\ee^{\lambda N_{\rm max}})\in O(1/\epsilon)$.  The
time required saturates the Heisenberg limit, up to a multiplicative constant.
Specifically, $\lambda\approx 0.17$ for RFPE.

\fig{ITPEcmp} compares RFPE to the information-theory PE (ITPE) method of Svore~\etal.  Although ITPE is inefficient and is non--adaptive, it is exact and so it is a natural benchmark to compare RFPE against.  We find  ITPE requires nearly five times the applications of $U$ if $\phi=2\pi k/t$ for integer $k<t$ and $t=10~000$.

Conversely,
ITPE requires only $25$ measurements to identify the phase with $50\%$
probability whereas RFPE requires $51$ experiments if $k$ is an integer.  If
the true value of $k$ is real-valued and ITPE is left unmodified, then  ITPE
fails to learn in the median because the long evolution times chosen lead to
contradictory possibilities which causes ITPE to become confused.  We correct
this by choosing $M\rightarrow \lceil M/2\rceil$ in such cases, which increases the number
of experiments to $35$ but also reduces the experimental time below that of
unmodified ITPE.  This shows that all three methods
tradeoff experimental and computational resources differently.  The latter point
is especially salient because ITPE is inefficient unlike RFPE.

\begin{figure*}
    \begin{centering}
        \includegraphics[width=0.4\linewidth]{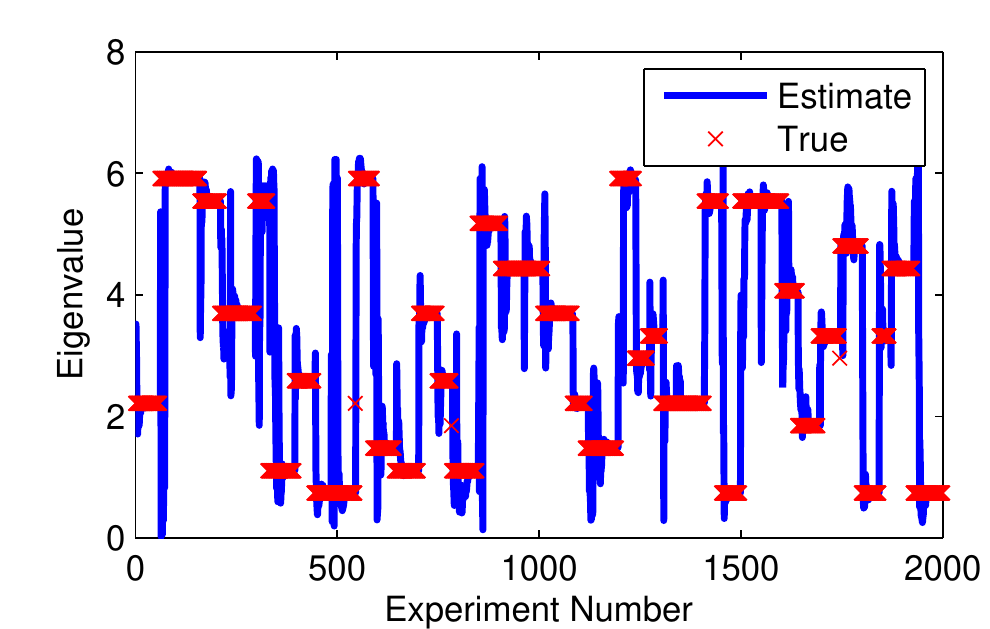}
        \hspace{5mm}
        \includegraphics[width=0.4\linewidth]{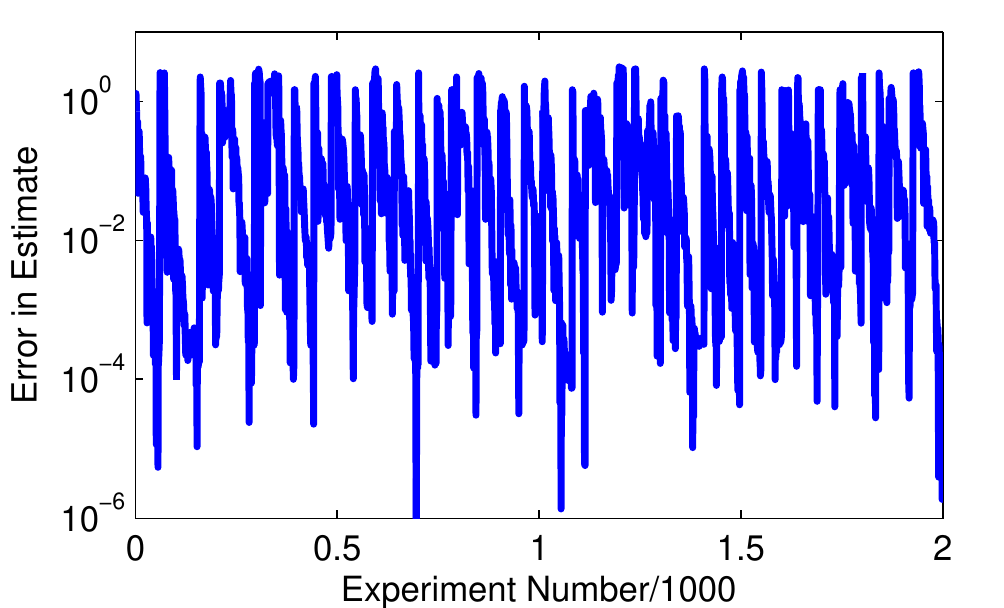}
    \end{centering}
    \caption{\label{fig:Errplot}
        Instantaneous estimate of the eigenphase  for a system with $16$ eigenvalues, $\Delta=0$, $\tau=0.1$ and $T_2=10^4$.
    }
\end{figure*}

\subsection{Phase estimation with depolarizing noise}
A criticism that has been levied lately at PE methods is that they can be impractical to execute on non--fault tolerant quantum hardware~\cite{PMS+14,MBL+14,WHT15}.  This is because phase estimation attains its quadratic advantage over statistical sampling by using exponentially long evolutions.  Decoherence causes the resultant phases to become randomized as time progresses, ultimately resulting in
$\lim_{M\rightarrow \infty} P(0|\phi;\theta,M) = \frac{1}{2}$.
In this limit, the measurements convey no information.  It similarly may be tempting to think that decoherence fundamentally places a lower limit on the accuracy of PE.   Bayesian inference to can, however, estimate $\phi$ using experiments with $M \approx T_2$~\cite{granade_robust_2012}.  

We model the effect of decoherence on the system by assuming the existence of a \emph{decoherence time} $T_2$ such that
\begin{gather}
    \label{eq:likedecohere}
    \begin{aligned}
        P(0|\phi) & = \ee^{-M/T_2}\left(\frac{1+\cos(M[\phi -\theta])}{2}\right)+\frac{1-\ee^{-M/T_2}}{2},\\
        P(1|\phi) & = \ee^{-M/T_2}\left(\frac{1-\cos(M[\phi -\theta])}{2}\right)+\frac{1-\ee^{-M/T_2}}{2}.
    \end{aligned}
\end{gather}
This model is appropriate when the time required to implement the controlled operation $\Lambda(U)$ is long relative the that required to perform $H$ and an arbitrary $Z$--rotation, as is appropriate in quantum simulation.  For simplicity, in the following we assume that $T_2$ is known.

Since $T_2$ places a limitation on our ability to learn, we propose a variation to \eq{PGH},
\begin{equation}
    \label{eq:pgh2}
    M = \min\left\{\left\lceil\frac{1.25}{\sigma}\right\rceil, T_2 \right\}.
\end{equation}
The experiments yielded by~\eq{pgh2} are qualitatively similar to locally optimized experiments for frequency estimation which are chosen to saturate the Cram\'er-Rao bound \cite{ferrie_how_2013}; however,~\eq{pgh2} requires nearly 100 fold less computing time to select an experiment.
We choose $M$ to cut off at $T_2$ in part because the locally optimized solutions choose to saturate at this value which can be understood using the following argument.  The Cram\'er-Rao bound for frequency estimation scales as $O(M^{-2})$~\cite{WGC15} in the absence of decoherence.  Eqn.~\eq{likedecohere} suggests that decoherence causes the posterior variance to increase as $O(\exp(2M/T_2))$. { As in the frequency estimation case \cite{ferrie_how_2013},} calculus reveals that $M=T_2$ optimally trades off these tendencies, which justifies~\eq{pgh2}.

\fig{T2plot} shows that RFPE smoothly transitions between the exponential scaling expected at short times and the polynomial scaling expected when decoherence becomes significant.  The error scales roughly as $1/N^{0.6}$ in this polynomial regime, which is comparable to the $1/\sqrt{N}$ scaling expected from statistical sampling.  Decoherence therefore does not necessarily impose a fundamental limitation on the accuracy that PE can achieve.  \fig{T2plot} also shows that confining the experiments to stay in a more coherent regime also can actually hurt the algorithm's ability to infer the phase, as expected.

\section{Tracking Eigenphases}

\fig{T2plot} shows the performance of RFPE when the initial quantum state is an eigenstate and is discarded after each experiment.  Performing phase estimation in this way minimizes the number of experiments, but can be prohibitively expensive if preparation of the initial eigenstate is prohibitively costly.  In such cases, it makes sense to follow the standard prescription for phase estimation by keeping the quantum state until it is clear that the initial eigenstate has been depolarized.  These depolarizations can cause RFPE to become confused because the new data that comes in is only consistent with hypotheses that have been ruled out.  We address this by performing inexpensive experiments to assess whether the state has depolarized and then restart the learning process. Restarting can even be valuable when $T_2=0$ if the input state is a superposition of eigenvectors or to recover if the RFPE becomes stuck (see \app{var-reduction}).

The following procedure addresses this issue in cases where the spectral gaps are promised to be at least $\Delta$.
\begin{enumerate}
\item After each update with probability $e^{-M/T_2}$ perform an experiment with $\theta=\mu$ and $M=\tau/\sigma$ for $\tau< 1$.
\item If result $=1$ then prepare initial state and reset $\sigma$.
\item After restart, continue as normal until $\sigma<\Delta$ then set $\sigma$ and $\mu$ to those of the closest eigenvalue.
\end{enumerate}

Steps 1 and 2 perform a one--sided test of whether the prior distribution is consistent with the current state.
If the prior probability distribution is correct then the probability of measuring $0$ is
\begin{equation}
    \frac{1}{\sigma\sqrt{2\pi}}\int_{-\infty}^\infty \cos^2\left(\frac{(\mu-x)\tau}{2\sigma}\right)\ee^{-\frac{(\mu-x)^2}{2\sigma^2}} \mathrm{d}x = \frac{1+\ee^{-\tau^2/2}}{2}.
\end{equation}
If $\tau=0.1$, the probability of measuring $0$ is approximately $0.998$ and hence measuring $1$ implies that the hypothesis that the prior is correct can be rejected at $p \le 0.002$. A Bayesian analysis of our reset rule
is given in \app{bf}.

Step 3 restarts the learning process if the learning process fails. It is important, however, to not throw away the spectral information that has been learned prior to the restart.  Step 3 reflects this by checking to see if the current estimate of the eigenphase corresponds to a known eigenstate and then sets $\mu$ to be the estimate of the eigenvalue and $\sigma$ to be its uncertainty if such an identification is made.  This allows the algorithm to resume learning once the depolarized state is projected onto a known eigenstate.

The test process also permits the eigenvalue of an eigenstate in a decohering system to be estimated in real time. \fig{Errplot} shows that the restarting algorithm can rapidly detect a transition away from the instantaneous eigenstate and then begin inferring the eigenvalue of the system's new instantaneous eigenstate.  

\section{Conclusion}

Our work makes phase estimation more
experimentally relevant by substantially reducing the
experimental time required and also by making the process resilient
to decoherence. This is especially critical as current experiments push past the classical
regime and IPE becomes increasingly impractical in lieu of fault-tolerance.
In particular, the ability of our algorithm to learn in the presence of decoherence provides an efficient alternative to the
variational eigensolvers used in present day experiments~\cite{PMS+14,MBL+14,WHT15}.  

Looking at the problem of phase estimation more generally, it is clear that it is in essence an inference problem.
Our work shows that modern ideas from machine learning can be used therein to great effect.  
It is our firm belief that by combining the tools of data science with clever experimental design,
further improvements can be seen not only in phase estimation but also quantum metrology in general.

\acknowledgments{We thank B. Terhal, K. Rudinger and D. Wecker for useful comments.}

%

\pagebreak
\appendix

\onecolumngrid

\section{Variance Reduction Strategies}
\label{app:var-reduction}

An important drawback of RFPE is  the tails for the distribution of errors in phase estimation can be quite fat in typical applications, meaning that there is significant probability that the error in the inferred eigenphase is orders of magnitude greater than the median.  Such a distribution can be seen in~\fig{PEerrorhist} where we see that although the median error is roughly $10^{-10}$ radians after $100$ updates for $m>50$  a non--negligible fraction of the experiments have error on the order of $1$.    

Fortunately, the need to always repeat the algorithm and use a majority voting scheme to reduce the variance of the estimate is mitigated by the fact that the algorithm outputs $\sigma$ which estimates the uncertainty in the resultant eigenphase. 
Alternatively, less aggressive experiments can be chosen in the guess heuristic or multi--modal models for the prior distribution (such as Gaussian mixture models) can be used.  The later approach can be quite valuable in computer vision and machine learning, however here we have an additional freedom not typically enjoyed in those fields: we can perform adaptive experiments to test to see if our current hypothesis is correct.

The central idea behind our restarting strategy is to examine the decay of $\sigma$ with the number of experiments.  In ideal cases, the error decays exponentially which means that it is easy to see when the algorithm fails by plotting $\sigma$ on a semi--log plot.  Such intuition can be easily automated.  The idea behind our restarting algorithm is to first estimate the derivative of $\log(\sigma)$ with respect to the experiment number and determine whether it is less than a threshold.  If it is less than the threshold, perform an experiment to test to see if the value of $\mu$ that has been learned so far is accurate (as per our incoherent phase estimation algorithm).  If it is found to be innacurate then restart the algorithm and abandon the information learned so far.  The restarting portion of this algorithm is given in~\alg{restart}.

When a restarting strategy like this is employed, we need to modify the model selection criteria.  Previously, we used the value of $\mu$ yielded by the most recent experiment as our estimate.   Given that a restart late in the algorithm can reset all information learned about a model, it makes sense in this context to use the value of $\mu$ corresponding to the smallest $\sigma$ observed.  This corresponds to the model that the inference algorithm has the greatest certainty.

We see in~\fig{restart} that this strategy of restarting substantially reduces the weights of the tails.  In fact, the mean error in the inference falls from $0.0513$ radians to $1.08\times 10^{-6}$ radians.  This shows that this resetting strategy substantially reduce the probability of a large error occuring in the estimate of the eigenphase.  To our knowledge, this data represents the first time that a heuristic method has been successful in reducing the mean error in frequency estimation to an acceptable level of uncertainty~\cite{granade_robust_2012}.  Previous approaches that succeeded in making the mean--error small used costly numerically optimized experiments, which render such approaches impractical for online phase estimation~\cite{granade_robust_2012,ferrie_how_2013, wiebe_hamiltonian_2014,wiebe_quantum_2014-1,WGC15}.

Note that Kitaev's phase estimation algorithm also can have a high probability of failure throughout the algorithm, so one might wonder if these restarting strategies might also be of value for traditional phase estimation algoirthms.  Unfortunately, these strategies cannot be easily translated to Kitaev's PE algorithm because the most expensive experiments are performed first.  Thus, restarting does not allow the algorithm to combat decoherence, nor can it be used to inexpensively fix situations where one of the bits is inferred improperly.

\begin{figure*}
    \begin{centering}
        \includegraphics[width=0.3\linewidth]{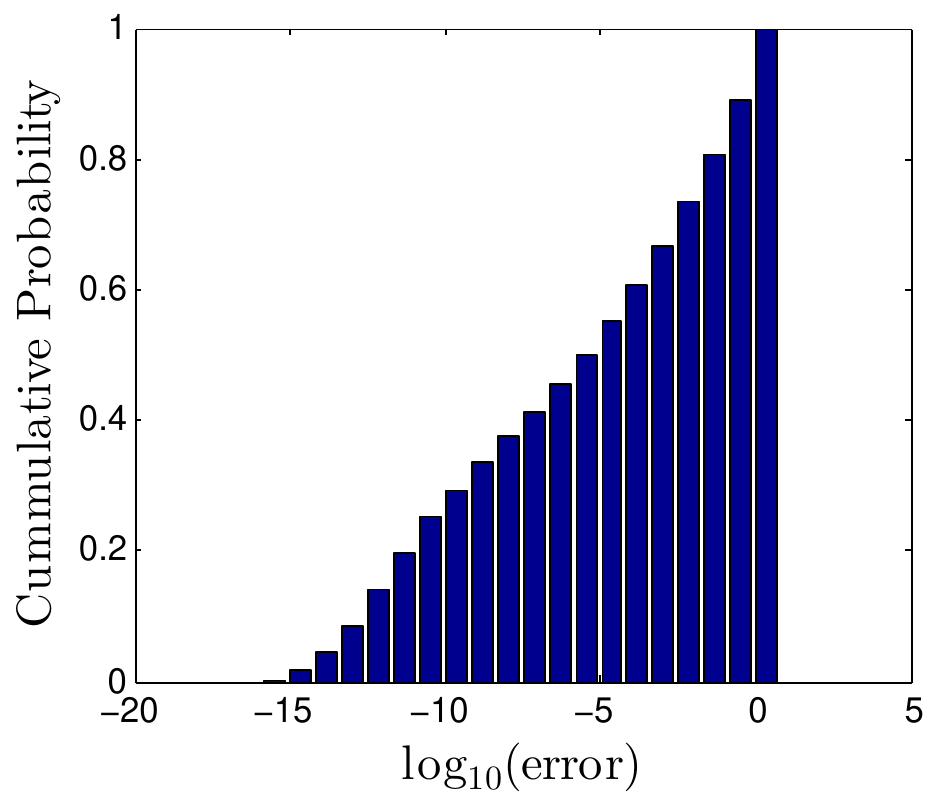}
        \includegraphics[width=0.3\linewidth]{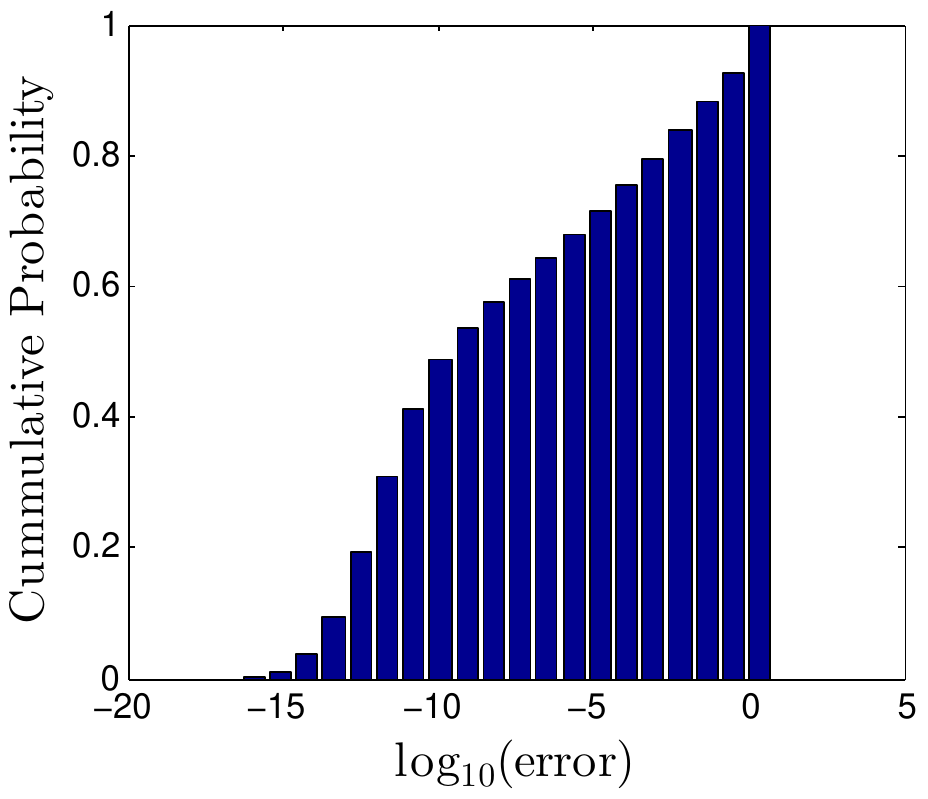}
        \includegraphics[width=0.3\linewidth]{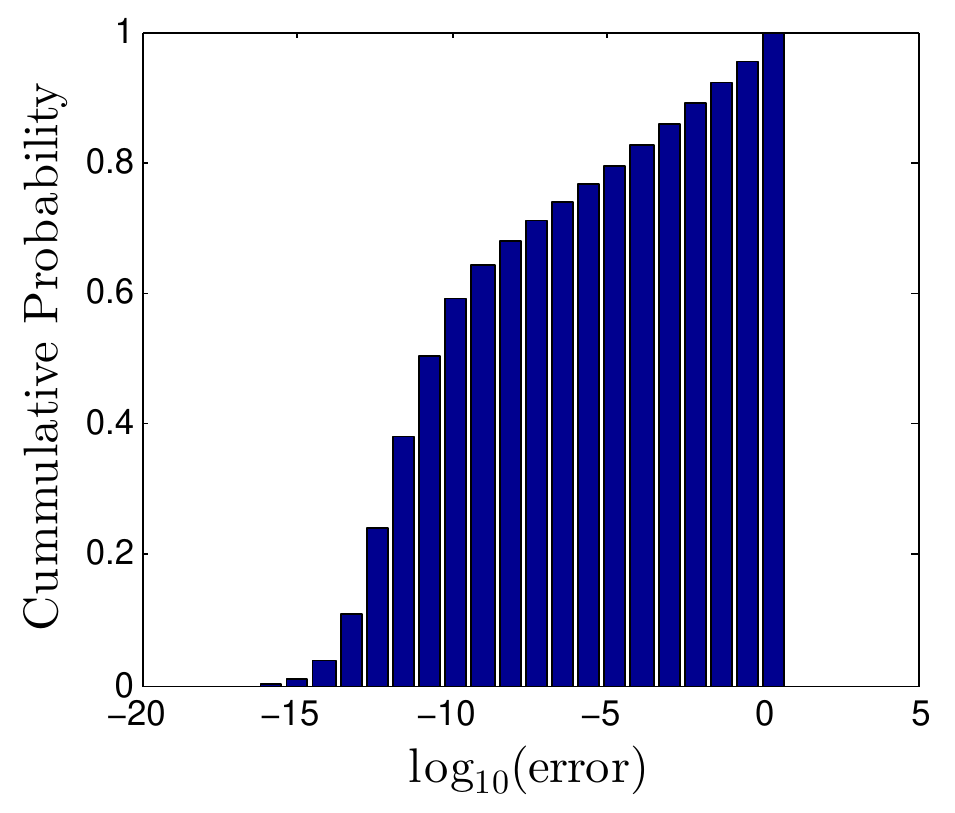}
    \end{centering}
    \caption{\label{fig:PEerrorhist}
     Cumulative distribution function of probability that PE error is less than $x$ after $150$ updates for $m=50$ (left) $m=100$ (middle) $m=200$ (right).
    }
\end{figure*}

\begin{figure*}
    \begin{centering}
        \includegraphics[width=0.35\linewidth]{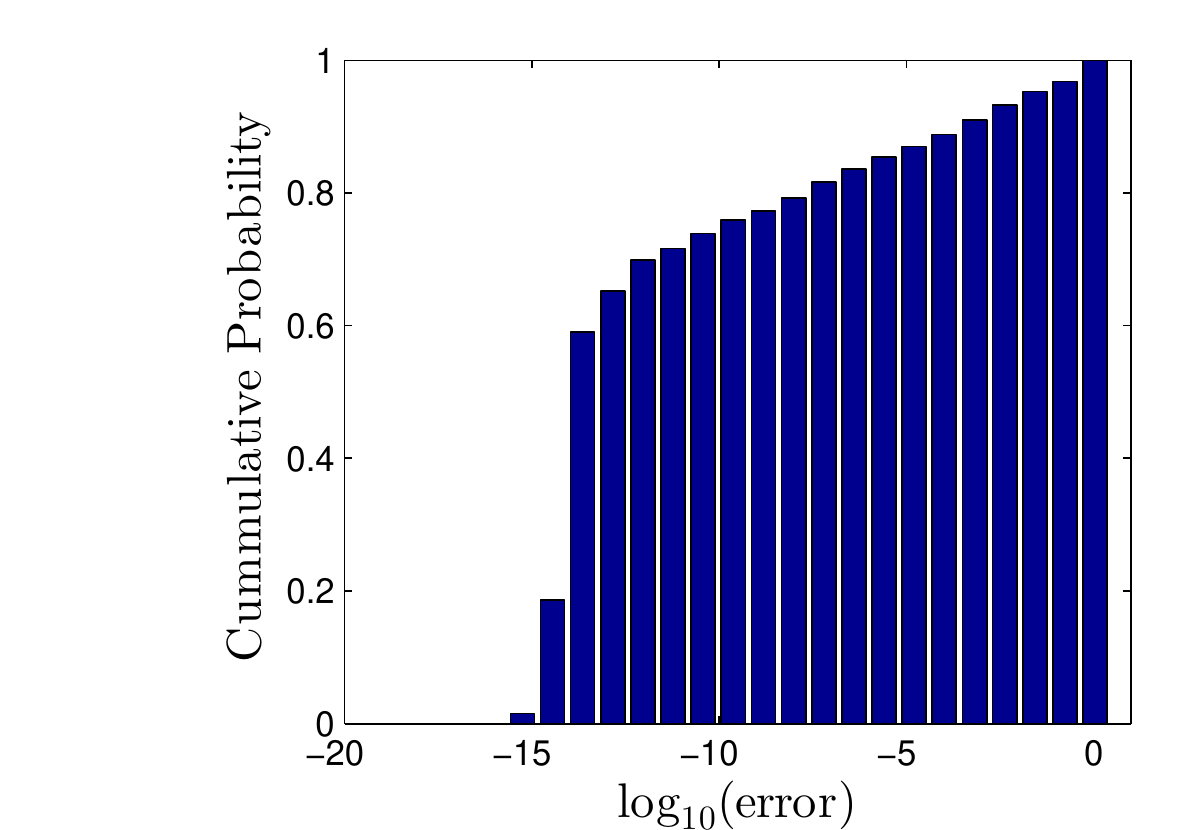}
        \includegraphics[width=0.35\linewidth]{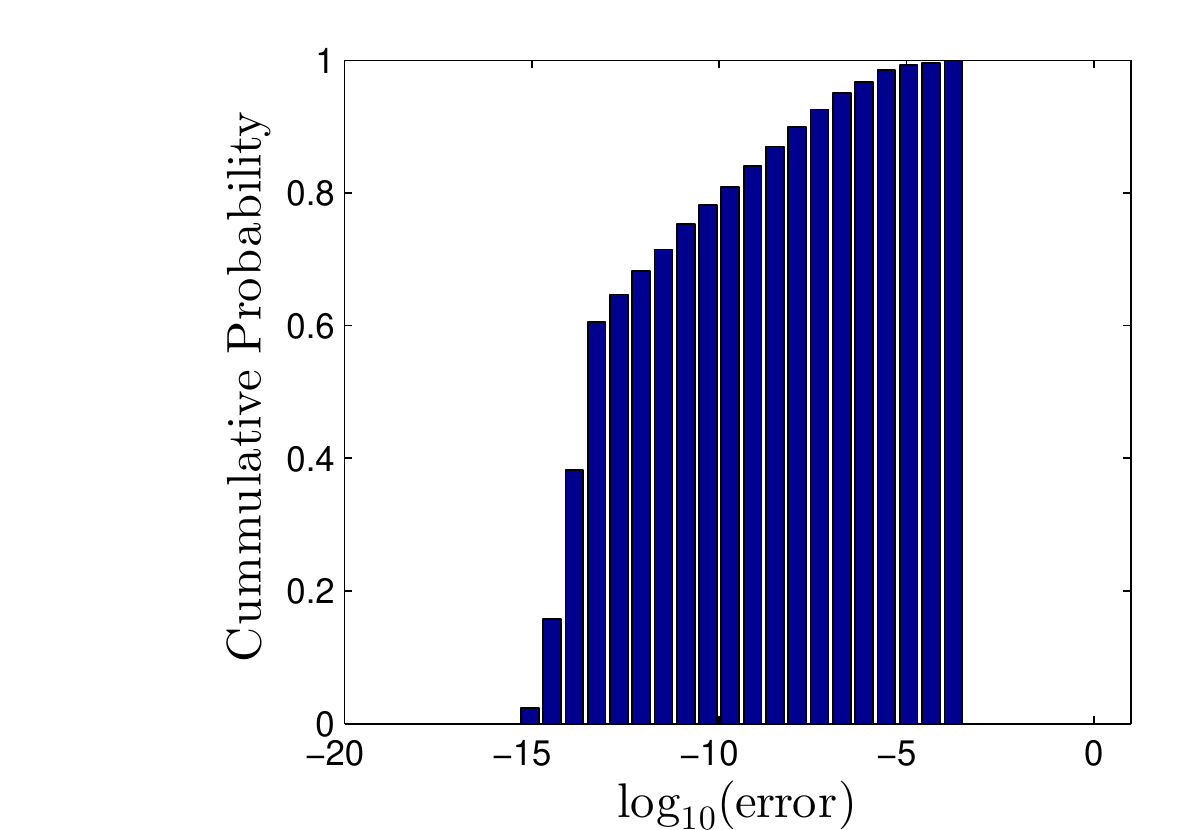}
    \end{centering}
    \caption{\label{fig:restart}
     Cumulative distribution function of probability that PE error is less than $x$ after $200$ updates for $m=2000$ for $\Gamma=\infty$ (left) $\Gamma=0.1$ (right) and $\tau=0.1$.  Estimate of CDF consists of $1000$ random eigenphase inference problems with $T_2=\infty$.
    }
\end{figure*}

\section{Stability Against Errors in the Likelihood Function}

Our algorithm is not only robust against well known depolarizing noise but is robust against~\emph{uncharacterized noise sources} as well.  We demonstrate this in~\fig{gamma} wherein we introduce depolarizing noise of strength $\gamma$ to our experimental system, but do not include such noise in the likelihood function.  For example, with $\gamma=0.4$, the measurement outcome of phase estimation is replaced with a random bit with probability $40\%$.  Although this may seem qualitatively similar to the numerical examples for depolarizing noise considered in the main body of the paper, this case is much more pathological because it is only used to generate the experimental data and is not permitted in the model of the system used in the likelihood function.  This raises the possibility that the algorithm could become confused due to experimental results that are fundamentally inconsistent with the assumed likelihood function for the system.

\fig{gamma} shows that the inclusion of uncharacterized depolarizing noise does not actually prevent the eigenvalue from being estimated.  Rather it reduces the number of bits per experiment that the algorithm can infer.  We see this by fitting the exponential part of the data in~\fig{gamma} (along with similar data for $\gamma=0.1$, $\gamma=0.3$, $\gamma=0.5$) and find that the error decay exponent, $\lambda$, to shrink as roughly $\lambda \approx 0.17e^{-3.1\gamma}$ it does not prevent our algorithm from learning at an exponential rate (until depolarizing noise becomes significant).  Thus RFPE continue to work even in the presence of depolarizing noise sources that are both strong and uncharacterized.

\section{Stability of Rejection Filtering PE}
\label{app:stability}

\begin{figure}
    \begin{centering}
\includegraphics[width=0.45\linewidth]{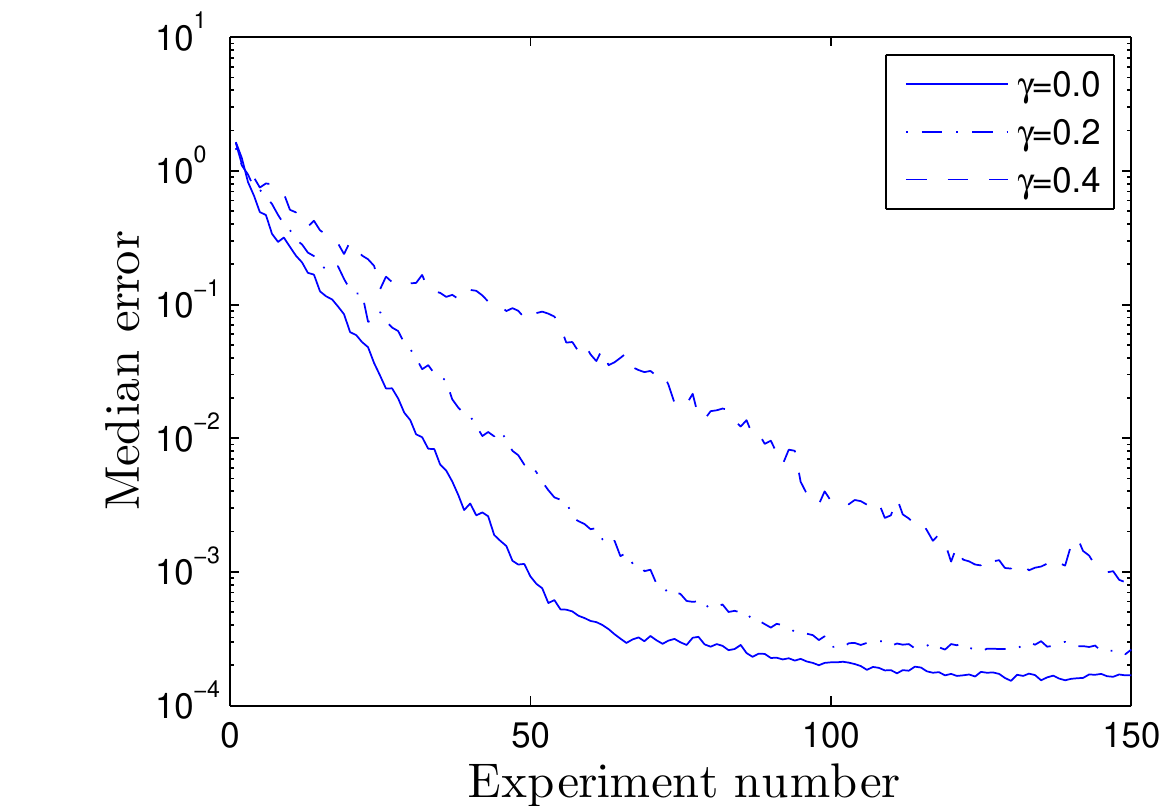}
    \end{centering}
    \caption{\label{fig:gamma}
Errors in inference for phase estimation with different levels of un-modeled noise, $\gamma$, for $T_2=1000$.  In each instance the initial state was taken to be a fixed, but unknown, eigenstate of the Hamiltonian.  The median was found for $100$ randomly chosen values of the eigenphase of this fixed initial eigenstate.
    }
\end{figure}

One way in which the rejection sampling algorithm can break down is if the likelihood function becomes too flat relative to the number of discrete samples, $m$, drawn from the prior distribution.  This breakdown occurs because the sample variance in the posterior mean is much greater than the difference between the prior mean and posterior means, which are not expected to be approximately the same if the likelihood function has little variation.  This in effect implies that the dynamics of the mean will essentially be a random walk and as a result we do not expect that our rejection sampling method to perform well if the likelihood function is too flat.  

In practice, the flatness of the distribution can be reduced by batching several experiments together and processing them.  This process is in general inefficient, but can be made efficient if an appropriate value of $\kappa_E$ is known for each datum recoreded.

The question remaining is: how flat is too flat?  We show in the theorem below that if the number of samples taken from the true posterior distribution does not scale at least inverse quadratically with the scale of relative fluctuations in the likelihood function then we expect the posterior variance to be much greater than the shift in the means.  
\begin{theorem}
Assume that for all $j$ $P(E|x_j) =\alpha+\delta_j$ where $|\delta_j|\le \delta$ and $\alpha \ge 10\delta$ and assume that $x_j\sim P(x_j)$.  If we then define $\mu_0 := \sum_{j} P(x_j) x_j$, $\mu_1:= \sum_j P(x_j|E) x_j$ and $\sigma_1^2$ to be the posterior variance then $|\mu_1 - \mu_0| \in \Omega(\sigma_1/\sqrt{m})$ only if
$
\frac{\alpha^2}{\delta^2}\in O(m) .
$\label{thm:stability}
\end{theorem}
\begin{proof}
Bayes' rule gives us
\begin{equation}
\left|\sum_j P(x_j|E) x_j -\mu_0\right|= \left|\frac{\sum_j P(E|x_j)P(x_j) (x_j -\mu_0)}{\sum_j P(E|x_j)P(x_j)}\right|=\left|\frac{\sum_j \delta_j P(x_j)(x_j -\mu_0)}{\sum_j P(E|x_j)P(x_j)}\right|.\label{eq:A1}
\end{equation}
Then using the Cauchy--Schwarz innequality, the triangle inequality and $\alpha\ge 2\delta$.
\begin{equation}
\left|\frac{\sum_j \delta_jP(x_j )( x_j -\mu_0)}{\sum_j P(E|x_j)P(x_j)}\right| \le \left|\frac{\delta \sqrt{\sum_j P(x_j) |x_j -\mu_0|^2}}{\alpha-\delta}\right|\le \frac{\delta  \sigma}{\alpha-\delta}\le \frac{2\delta{\sigma}}{\alpha}.\label{eq:A2}
\end{equation}
Thus the maximum shift in the posterior mean shrinks as the the likelihood function becomes increasingly flat, as expected.

Next we need to lower bound the posterior variance in order ensure that the value of $m$ chosen suffices to make the error small in the best possible case.
To do so we use the reverse triangle inequality:
\begin{equation}
\sigma_1^2 = |\sigma_1^2 -\sigma^2 +\sigma^2| \ge \sigma^2 - |\sigma_1^2-\sigma^2|.
\end{equation}
Thus it suffices to upper bound $|\sigma_1^2-\sigma^2|$ to lower bound $\sigma_1^2$.  To do so, note that $\alpha \ge 2\delta$ and hence

\begin{align}
|P(x_j|E)-P(x_j)| = \left|\frac{P(x_j)(\delta_j-\sum_jP(x_j)\delta_j)}{\alpha+\sum_j P(x_j)\delta_j}\right|\le \frac{P(x_j) 2\delta}{\alpha-\delta}\le \frac{P(x_j) 4\delta}{\alpha}.
\end{align}
Now the difference between the two variances can be written as
\begin{align}
|\sigma_1^2-\sigma^2| &= \left|\sum_j  P(x_j|E)(x_j-\mu_1)^2-P(x_j)(x_j-\mu_0)^2\right|\nonumber\\
 &\le \left|\sum_j  (P(x_j|E)-P(x_j))(x_j-\mu_1)^2\right|+\left|\sum_j P(x_j)\left((x_j-\mu_0)^2-(x_j-\mu_1)^2\right)\right|\nonumber\\
 &\le  \frac{4\delta}{\alpha}\left|\sum_j  P(x_j)(x_j-\mu_1)^2\right|+(\mu_1-\mu_0)^2\nonumber\\
&\le \frac{4\delta\sigma^2}{\alpha}+\frac{4\delta}{\alpha}\left|\sum_j P(x_j)[(x_j-\mu_1)^2-(x_j-\mu_0)^2]  \right|+(\mu_1-\mu_0)^2\nonumber\\
&\le \frac{4\delta\sigma^2}{\alpha}+(1+\frac{4\delta}{\alpha})(\mu_1-\mu_0)^2\le \frac{4\delta\sigma^2}{\alpha}+ \frac{12\delta^2\sigma^2}{\alpha^2}\le \frac{10\delta\sigma^2}{\alpha}.
\end{align}
Thus we have that
\begin{equation}
\sigma_1^2 \ge \sigma^2(1-10\delta/\alpha).
\end{equation}
Now assuming $\delta\le \alpha/10$ we have 
\begin{equation}
\sigma_1^2\in \Omega(\sigma^2).
\end{equation}
Finally, we note that
\begin{equation}
|\mu_1-\mu_0| \in \Omega(\sigma_1/\sqrt{m}) \Rightarrow \frac{\delta \sigma}{\alpha} \in \Omega(\sigma/\sqrt{m}),
\end{equation}
which is only true if $m\in \Omega(\alpha^2/\delta^2)$.
\end{proof}
We therefore see that the number of samples needed to track the small changes in a posterior distribution that happens when the likelihood function becomes extremely flat.  This condition is not sufficient because the actual components of the posterior mean may be shifted by a much smaller amount than the upper bounds used in the proof of~\thm{stability}.

In contrast, exact Bayesian inference requires a number of bits that scales as $O(\log(1/\delta))$ (assuming a fixed and discrete number of hypotheses).  Thus exact Bayesian inference (or to a lesser extent particle filter methods) can be preferable in cases where the likelihood function is extremely flat.  Such concerns can be somewhat be either be avoided entirely or batches of such experiments should be combined to produce a likelihood function that is much less flat and choosing an appropriate instrumental distribution to ensure that the success probability remains high.

\section{Bayes Factors for Reset Rule}
\label{app:bf}

Though in the main body, we have chosen to present the reset step in terms of
$p$-values for familiarity, we note that $p$-values are difficult to correctly
interpret and can lead to misunderstandings \cite{goodman_dirty_2008,hoekstra_robust_2014}. As an
alternative, one may prefer to use the Bayes factor to test whether the rejection
filter has failed. For example, the rejection filter could fail due to
a failure of the numerical approximation or because the eigenstate has
been lost, as described in the main body. For a uniform prior over
whether the rejection filter has failed, this reduces to the likelihood
ratio test
\begin{subequations}
    \begin{align}
        L & = \frac{\Pr(\text{result} = 1 | \text{prior wrong})}{\Pr(\text{result} = 1 | \text{prior correct})} \\
          & = \frac{                  
                  1 - \ee^{
                          - (\tau^2 \sigma_\reset^2 / 2 \sigma^2 + \sigma \tau / T_2)
                      }
                      \cos \left(
                        \tau \left(\mu -\mu _\reset \right) / \sigma
                      \right)
              }{
                  1-\ee^{-\tau^2 / 2}
              }
    \end{align}
\end{subequations}
where $\mu_\reset$ and $\sigma_\reset$ are the values of $\mu$ and $\sigma$
immediately following a reset. Using this test, the degree by which $L > 1$
informs as to the degree to which we should prefer the model proposed by the
reset rule. Again under the assumption of a uniform prior over the
validity of the reset rule,
\begin{equation}
    \Pr(\text{prior wrong} | \text{result} = 1) = L \Pr(\text{prior right} | \text{result} = 1).
\end{equation}
For instance, if the variance has been reduced by a factor of 100 from
its initial value ($\sigma = \sigma_\reset / 100$) and the current mean is
correct ($\mu = \mu_\reset$), then assuming $T_2 = 100$ and an initial standard
deviation of $\sigma_\reset = \pi / \sqrt{3}$, $L\approx8000$ for a result of
1. That is, the initial prior is 8,000 times as probably correct as the current
prior in this example.

\begin{figure}
    \begin{center}
        \includegraphics[width=0.7\textwidth]{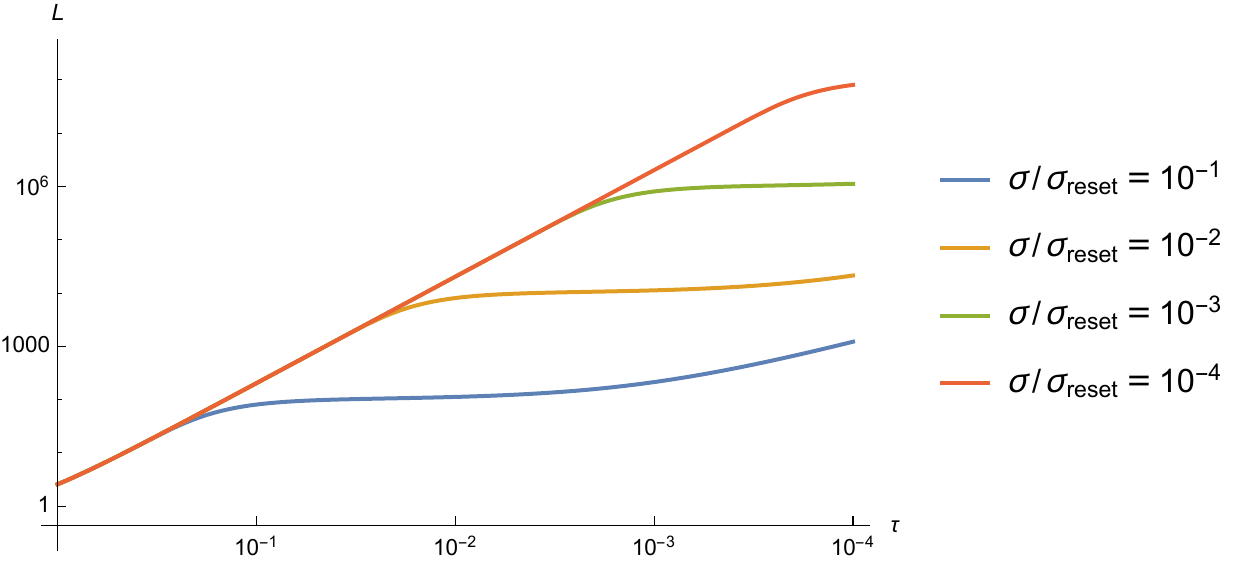}
    \end{center}
    \caption{
        \label{fig:reset-bf-thresholds}
        Likelihood ratio test values for various settings of the parameter
        $\tau$, and for various prior variances $\sigma / \sigma_\reset$.
        In this example, we suppose that $\mu = \mu_\reset$ and that $T_2 = 100$.
    }
\end{figure}

Importantly, this example rests on the assumption that one take a uniform
prior over whether the current prior is valid. This assumption corresponds to
that which an impartial observer might take in evaluating whether the numerical
approximations made by our rejection filter algorithm have failed for the
observed data. That is, the reset rule proposed corresponds performing an intervention
without relying on the full experimental data record. Choosing a threshold
other than $L = 1$ represents placing a different prior, as could be informed by observing
the failure modalities of many runs of the rejection filter method. As
demonstrated in \fig{reset-bf-thresholds}, because our reset rule resets with
probability 1 if a 1 is observed, choosing $\tau$ effectively sets the threshold
for $L$.
In practice,
however, because $\Pr(\text{result} = 1 | \text{prior correct}) \approx 0$, the reset
rule is only weakly dependent on the specific threshold one places on $L$.

\section{Pseudocode for Algorithms}
\label{app:pseudocode}

In the main body we sketched the details of our phase estimation algorithm.  Here we elaborate on this algorithm and discuss some of the
subtle details needed to make the algorithm work.  The first such subtlety stems from the fact that eigenphases are equivalent modulo $2\pi$.  
To see this problem, consider a Gaussian distribution centered at $0$.  If we take the outputs of the distribution in the branch $[0,2\pi]$ then we find that the mean of the distribution is $\pi$ rather than $0$.  Since the support of the initial Gaussian may be small at $\phi=\pi$, such errors can be catastrophic during the inference procedure.  This can be dealt by using the circular mean and by working with a wrapped normal distribution.  This approach is discussed in~\alg{crej2}.  For expedience, we eschew this approach and instead use a heuristic approach that does  not require a large number of trigonometric calls, which can be expensive if the device performing the inference does not have native trigonometric calls.

The heuristic approach, described in~\alg{crej}, uses rejection sampling and incremental methods to estimate the mean and standard deviation of the posterior distribution.  If $\sigma\ll 2\pi$ then the probability distribution is narrow and does not suffer from significant wrap around unless $\mu \mod 2\pi \approx 0$.  We address this by keeping track of each of the accepted $\phi_j$ as well as $\phi_j+\pi$.  If $\sigma\ll 1$ than it is impossible that both distributions suffer from substantial wrap around.  The arithmetic, rather than circular, mean and standard deviation are then computed for both using an incremental formula and the branch with the smallest variance is kept.  If the branch that was shifted by $\pi$ is chosen, then $\pi$ is subtracted from the reported mean.  The standard deviation does not need to be modified because it is invariant with respect to displacements of the mean.

While this approach is correct if $\sigma\ll 1$, it is only approximate if $\sigma$ is on the order of $1$.  In such cases, computation of the circular mean is much better justified, however we find that using our heuristic approach continues to provide acceptable results while avoiding trigonometric calls that can be expensive in some contexts.  An alternative approach to solving this problem is to begin each phase estimation run with a batch of random experiments, as per~\cite{SHF14}, before continuing to ensure that the posterior variance is small enough to neglect the wrap around effect.

The choice of the evolution time and the inversion angle strongly impacts the efficiency of the learning algorithm.  We provide below code for a modified version of the
particle guess heuristic of~\cite{wiebe_hamiltonian_2014}.  As discussed in the main body, we expect that choosing $M> T_2$ will typically lead to worse estimates of the eigenphase because the effects of decoherence overwhelm the information that can be gleaned from these long experiments.  As a result, we modify the particle guess heuristic  to never choose $M> T_2$.  We formally state this procedure in~\alg{pghT2}.

\begin{figure}[h]
\begin{algorithm}[H]
    \caption{Bayes Update for \CRej using Directional Statistics}
    \label{alg:crej2}
    \begin{algorithmic}

        \Require Prior mean and variance $\mu,\sigma$, measurement $E$,
            settings $M,\theta$, number of attempts $m$, scale $\kappa_E$

        \linecomment{Initialize accumulators to 0.}
	\State $(x_{\rm inc},y_{\rm inc},N_a) \gets 0$.
        \linecomment{Attempt each sample.}

        \For{$i \in 1 \to m$}
            \State $x \sim\frac{e^{-(\phi-\mu)^2/2 \sigma^2}}{\sigma\sqrt{2 \pi }},$
          
            \linecomment{Accept or reject the new sample.}
            \State $u \sim \operatorname{Uniform}(0, 1)$
            \If{$P(E | x) \ge \kappa_Eu$}
                \linecomment{Accumulate using the accepted sample using Cartesian coordinates.}
                \State $x_{\rm inc} \gets x_{\rm inc}+ \cos(x)$
                \State $y_{\rm inc} \gets y_{\rm inc}+ \sin(x)$
                \State $N_a \gets N_a +1$.
            \EndIf
        \EndFor

        \State $x_{\rm inc}\gets x_{\rm inc}/N_a $
        \State $y_{\rm inc}\gets y_{\rm inc}/N_a $

        \linecomment{Return mean, variance of the posterior using accumulators.}
	\State $\mu\gets {\rm Arg}(x_{\rm inc}+iy_{\rm inc})$.
\linecomment{Use circular standard deviation to estimate SD for wrapped Gaussian}
	\State $\sigma = \sqrt{\ln\left(\frac{1}{\sqrt{x_{\rm inc}^2 + y_{\rm inc}^2}}\right)}$
        \State\Return $(\mu,\sigma)$

    \end{algorithmic}
\end{algorithm}
\end{figure}

\begin{figure}[h]
\begin{algorithm}[H]
    \caption{Bayes Update for \CRej}
    \label{alg:crej}
    \begin{algorithmic}

        \Require Prior mean and variance $\mu,\sigma$, measurement $E$,
            settings $M,\theta$, number of attempts $m$, scale $\kappa_E$

        \linecomment{Initialize accumulators to 0.}
        \State{$(\mu_{\rm inc},\mu_{\rm inc}',V_{\rm inc},V_{\rm inc}',N_a) \gets 0$}
        \linecomment{Attempt each sample.}
        \For{$i \in 1 \to m$}
            \linecomment{Draw a sample using each ``cut'' of the prior.}
            \State $x \sim\frac{e^{-(\phi-\mu)^2/2 \sigma^2}}{\sigma\sqrt{2 \pi }},$
            \State $x\gets x \mod 2 \pi$.
            \State $x'\gets x+\pi \mod 2 \pi$.

            \linecomment{Accept or reject the new sample.}
            \State $u \sim \operatorname{Uniform}(0, 1)$
            \If{$P(E | x) \ge \kappa_Eu$}
                \linecomment{Accumulate using the accepted sample w/ each ``cut.''}
                \State $\mu_{\rm inc} \gets \mu_{\rm inc}+ x$
                \State $V_{\rm inc} \gets V_{\rm inc}+ x^2$
                \State $V_{\rm inc}' \gets V_{\rm inc}'+ x'^2$
                \State $N_a \gets N_a +1$.
            \EndIf
        \EndFor
        \linecomment{Return mean, variance of the posterior using accumulators.}
        \State $\mu'\gets \mu_{\rm inc}/N_a $
        \State $\sigma' \gets \min\left(\sqrt{\frac{1}{N_a -1}\left(V_{\rm inc} - \mu_{\rm inc}^2 \right)},\sqrt{\frac{1}{N_a -1}\left(V_{\rm inc}' - \mu_{\rm inc}'^2 \right)}\right)$
        \State\Return $(\mu',\sigma')$

    \end{algorithmic}
\end{algorithm}
\end{figure}

\begin{figure}
\begin{algorithm}[H]
    \caption{Restarting algorithm}
    \label{alg:restart}
\begin{algorithmic}
        \Require Prior \CRej state, records of all previous models found in the phase estimation algorithm $\mu$, $\vec{\sigma}$, initial standard deviation $\sigma_{\rm init}$, $M$, $T_2$, a counter $\text{CNT}$, $\Gamma$ and $\tau$.
	\Ensure $\text{CNT}$, $\sigma$ 
        \Function{$\text{Restart}$}{${\mu}$, $\vec \sigma$,$\sigma_{\rm init}$, $M$, $T_2$, $\text{CNT}$, $\Gamma$, $\tau$}
	\State $D \gets$ derivative of $\log{\sigma}$.
	\If {$D\ge \Gamma$ and $\text{CNT}<5$ or rand()$>\exp(-M/T_2)$}\Comment{Checks to see if the eigenstate is suspect.}
		\State Perform experiment with $M=\tau/\sigma$ and $\theta=\mu$.
		\If{Outcome is 0}\Comment{Test concludes state estimate is valid}
\State ${\rm CNT} \gets {\rm CNT}+1$.
			\State \Return $\text{CNT},\sigma$.
\Else\Comment{Test concludes state estimate is invalid}

\State $\text{CNT}\gets 0$
	\State $\sigma \gets \sigma_{\rm init}$
	\State \Return $\text{CNT}, \sigma$
\EndIf
\Else\Comment{Does not restart if state is not suspect}
	\State $\text{CNT}\gets \text{CNT}+1$
		\State \Return $\text{CNT},\sigma$.
	\EndIf
        \EndFunction
    \end{algorithmic}
\end{algorithm}
\end{figure}

\begin{figure}
\begin{algorithm}[H]
    \caption{PGH for decoherent phase estimation using \CRej}
    \label{alg:pghT2}
\begin{algorithmic}
        \Require Prior \CRej state $\mu$, $\Sigma$. Resampling kernel $\operatorname{F}$.
        \Ensure  An experiment $(M, \theta)$.
        \Function{$\text{PGH}_\text{\CRej}$}{$\mu$, $\Sigma$, $T_2$}
            \State $M \gets 1.25 / \sqrt{{\rm Tr}(\Sigma)}$
    \If {$M\ge T_2$}
        \State $M\sim f(x;1/T_2)$\Comment{Draw $M$ from an exponential distribution with mean $T_2$}.
    \EndIf
            \State $(-\theta/M) \sim \operatorname{F}(\mu, \Sigma)$
            \State \Return $(M, \theta)$.
        \EndFunction
    \end{algorithmic}
\end{algorithm}
\end{figure}


\begin{thebibliography}{27}%
\makeatletter
\providecommand \@ifxundefined [1]{%
 \@ifx{#1\undefined}
}%
\providecommand \@ifnum [1]{%
 \ifnum #1\expandafter \@firstoftwo
 \else \expandafter \@secondoftwo
 \fi
}%
\providecommand \@ifx [1]{%
 \ifx #1\expandafter \@firstoftwo
 \else \expandafter \@secondoftwo
 \fi
}%
\providecommand \natexlab [1]{#1}%
\providecommand \enquote  [1]{``#1''}%
\providecommand \bibnamefont  [1]{#1}%
\providecommand \bibfnamefont [1]{#1}%
\providecommand \citenamefont [1]{#1}%
\providecommand \href@noop [0]{\@secondoftwo}%
\providecommand \href [0]{\begingroup \@sanitize@url \@href}%
\providecommand \@href[1]{\@@startlink{#1}\@@href}%
\providecommand \@@href[1]{\endgroup#1\@@endlink}%
\providecommand \@sanitize@url [0]{\catcode `\\12\catcode `\$12\catcode
  `\&12\catcode `\#12\catcode `\^12\catcode `\_12\catcode `\%12\relax}%
\providecommand \@@startlink[1]{}%
\providecommand \@@endlink[0]{}%
\providecommand \url  [0]{\begingroup\@sanitize@url \@url }%
\providecommand \@url [1]{\endgroup\@href {#1}{\urlprefix }}%
\providecommand \urlprefix  [0]{URL }%
\providecommand \Eprint [0]{\href }%
\providecommand \doibase [0]{http://dx.doi.org/}%
\providecommand \selectlanguage [0]{\@gobble}%
\providecommand \bibinfo  [0]{\@secondoftwo}%
\providecommand \bibfield  [0]{\@secondoftwo}%
\providecommand \translation [1]{[#1]}%
\providecommand \BibitemOpen [0]{}%
\providecommand \bibitemStop [0]{}%
\providecommand \bibitemNoStop [0]{.\EOS\space}%
\providecommand \EOS [0]{\spacefactor3000\relax}%
\providecommand \BibitemShut  [1]{\csname bibitem#1\endcsname}%
\let\auto@bib@innerbib\@empty
\bibitem [{\citenamefont {Shor}(1997)}]{shor_polynomial-time_1995}%
  \BibitemOpen
  \bibfield  {author} {\bibinfo {author} {\bibfnamefont {P.~W.}\ \bibnamefont
  {Shor}},\ }\bibfield  {title} {\enquote {\bibinfo {title} {Polynomial-time
  algorithms for prime factorization and discrete logarithms on a quantum
  computer},}\ }\href {\doibase 10.1137/S0097539795293172} {\bibfield
  {journal} {\bibinfo  {journal} {SIAM Journal on Computing}\ }\textbf
  {\bibinfo {volume} {26}},\ \bibinfo {pages} {1484} (\bibinfo {year}
  {1997})}\BibitemShut {NoStop}%
\bibitem [{\citenamefont {Brassard}\ \emph {et~al.}(2002)\citenamefont
  {Brassard}, \citenamefont {Hoyer}, \citenamefont {Mosca},\ and\ \citenamefont
  {Tapp}}]{BHM+02}%
  \BibitemOpen
  \bibfield  {author} {\bibinfo {author} {\bibfnamefont {G.}~\bibnamefont
  {Brassard}}, \bibinfo {author} {\bibfnamefont {P.}~\bibnamefont {Hoyer}},
  \bibinfo {author} {\bibfnamefont {M.}~\bibnamefont {Mosca}}, \ and\ \bibinfo
  {author} {\bibfnamefont {A.}~\bibnamefont {Tapp}},\ }\bibfield  {title}
  {\enquote {\bibinfo {title} {Quantum amplitude amplification and
  estimation},}\ }\href {\doibase 10.1090/conm/305} {\bibfield  {journal}
  {\bibinfo  {journal} {Contemporary Mathematics}\ }\textbf {\bibinfo {volume}
  {305}},\ \bibinfo {pages} {53} (\bibinfo {year} {2002})},\ \Eprint
  {http://arxiv.org/abs/quant-ph/0005055} {quant-ph/0005055} \BibitemShut
  {NoStop}%
\bibitem [{\citenamefont {Aspuru-Guzik}\ \emph {et~al.}(2005)\citenamefont
  {Aspuru-Guzik}, \citenamefont {Dutoi}, \citenamefont {Love},\ and\
  \citenamefont {Head-Gordon}}]{ADL+05}%
  \BibitemOpen
  \bibfield  {author} {\bibinfo {author} {\bibfnamefont {A.}~\bibnamefont
  {Aspuru-Guzik}}, \bibinfo {author} {\bibfnamefont {A.~D.}\ \bibnamefont
  {Dutoi}}, \bibinfo {author} {\bibfnamefont {P.~J.}\ \bibnamefont {Love}}, \
  and\ \bibinfo {author} {\bibfnamefont {M.}~\bibnamefont {Head-Gordon}},\
  }\bibfield  {title} {\enquote {\bibinfo {title} {Simulated quantum
  computation of molecular energies},}\ }\href {\doibase
  10.1126/science.1113479} {\bibfield  {journal} {\bibinfo  {journal}
  {Science}\ }\textbf {\bibinfo {volume} {309}},\ \bibinfo {pages} {1704}
  (\bibinfo {year} {2005})}\BibitemShut {NoStop}%
\bibitem [{\citenamefont {Harrow}\ \emph {et~al.}(2009)\citenamefont {Harrow},
  \citenamefont {Hassidim},\ and\ \citenamefont {Lloyd}}]{harrow2009quantum}%
  \BibitemOpen
  \bibfield  {author} {\bibinfo {author} {\bibfnamefont {A.~W.}\ \bibnamefont
  {Harrow}}, \bibinfo {author} {\bibfnamefont {A.}~\bibnamefont {Hassidim}}, \
  and\ \bibinfo {author} {\bibfnamefont {S.}~\bibnamefont {Lloyd}},\ }\bibfield
   {title} {\enquote {\bibinfo {title} {Quantum algorithm for linear systems of
  equations},}\ }\href {\doibase 10.1103/PhysRevLett.103.150502} {\bibfield
  {journal} {\bibinfo  {journal} {Physical Review Letters}\ }\textbf {\bibinfo
  {volume} {103}},\ \bibinfo {pages} {150502} (\bibinfo {year}
  {2009})}\BibitemShut {NoStop}%
\bibitem [{\citenamefont {Lanyon}\ \emph {et~al.}(2010)\citenamefont {Lanyon},
  \citenamefont {Whitfield}, \citenamefont {Gillett}, \citenamefont {Goggin},
  \citenamefont {Almeida}, \citenamefont {Kassal}, \citenamefont {Biamonte},
  \citenamefont {Mohseni}, \citenamefont {Powell}, \citenamefont {Barbieri}
  \emph {et~al.}}]{lanyon2010towards}%
  \BibitemOpen
  \bibfield  {author} {\bibinfo {author} {\bibfnamefont {B.~P.}\ \bibnamefont
  {Lanyon}}, \bibinfo {author} {\bibfnamefont {J.~D.}\ \bibnamefont
  {Whitfield}}, \bibinfo {author} {\bibfnamefont {G.}~\bibnamefont {Gillett}},
  \bibinfo {author} {\bibfnamefont {M.~E.}\ \bibnamefont {Goggin}}, \bibinfo
  {author} {\bibfnamefont {M.~P.}\ \bibnamefont {Almeida}}, \bibinfo {author}
  {\bibfnamefont {I.}~\bibnamefont {Kassal}}, \bibinfo {author} {\bibfnamefont
  {J.~D.}\ \bibnamefont {Biamonte}}, \bibinfo {author} {\bibfnamefont
  {M.}~\bibnamefont {Mohseni}}, \bibinfo {author} {\bibfnamefont {B.~J.}\
  \bibnamefont {Powell}}, \bibinfo {author} {\bibfnamefont {M.}~\bibnamefont
  {Barbieri}},  \emph {et~al.},\ }\bibfield  {title} {\enquote {\bibinfo
  {title} {Towards quantum chemistry on a quantum computer},}\ }\href {\doibase
  10.1038/nchem.483} {\bibfield  {journal} {\bibinfo  {journal} {Nature
  Chemistry}\ }\textbf {\bibinfo {volume} {2}},\ \bibinfo {pages} {106}
  (\bibinfo {year} {2010})}\BibitemShut {NoStop}%
\bibitem [{\citenamefont {Kitaev}(1996)}]{Kit96}%
  \BibitemOpen
  \bibfield  {author} {\bibinfo {author} {\bibfnamefont {A.~Y.}\ \bibnamefont
  {Kitaev}},\ }\bibfield  {title} {\enquote {\bibinfo {title} {Quantum
  measurements and the {A}belian stabilizer problem},}\ }\href
  {http://eccc.hpi-web.de/report/1996/003/} {\bibfield  {journal} {\bibinfo
  {journal} {Electronic Colloquium on Computational Complexity}\ }\textbf
  {\bibinfo {volume} {3}} (\bibinfo {year} {1996})}\BibitemShut {NoStop}%
\bibitem [{\citenamefont {Kitaev}\ \emph {et~al.}(2002)\citenamefont {Kitaev},
  \citenamefont {Shen},\ and\ \citenamefont {Vyalyi}}]{kitaev2002classical}%
  \BibitemOpen
  \bibfield  {author} {\bibinfo {author} {\bibfnamefont {A.~Y.}\ \bibnamefont
  {Kitaev}}, \bibinfo {author} {\bibfnamefont {A.}~\bibnamefont {Shen}}, \ and\
  \bibinfo {author} {\bibfnamefont {M.~N.}\ \bibnamefont {Vyalyi}},\
  }\href@noop {} {\emph {\bibinfo {title} {Classical and quantum
  computation}}},\ Vol.~\bibinfo {volume} {47}\ (\bibinfo  {publisher}
  {American Mathematical Society Providence},\ \bibinfo {year}
  {2002})\BibitemShut {NoStop}%
\bibitem [{\citenamefont {Higgins}\ \emph {et~al.}(2007)\citenamefont
  {Higgins}, \citenamefont {Berry}, \citenamefont {Bartlett}, \citenamefont
  {Wiseman},\ and\ \citenamefont {Pryde}}]{higgins2007entanglement}%
  \BibitemOpen
  \bibfield  {author} {\bibinfo {author} {\bibfnamefont {B.~L.}\ \bibnamefont
  {Higgins}}, \bibinfo {author} {\bibfnamefont {D.~W.}\ \bibnamefont {Berry}},
  \bibinfo {author} {\bibfnamefont {S.~D.}\ \bibnamefont {Bartlett}}, \bibinfo
  {author} {\bibfnamefont {H.~M.}\ \bibnamefont {Wiseman}}, \ and\ \bibinfo
  {author} {\bibfnamefont {G.~J.}\ \bibnamefont {Pryde}},\ }\bibfield  {title}
  {\enquote {\bibinfo {title} {Entanglement-free {H}eisenberg-limited phase
  estimation},}\ }\href {\doibase 10.1038/nature06257} {\bibfield  {journal}
  {\bibinfo  {journal} {Nature}\ }\textbf {\bibinfo {volume} {450}},\ \bibinfo
  {pages} {393} (\bibinfo {year} {2007})}\BibitemShut {NoStop}%
\bibitem [{\citenamefont {Svore}\ \emph {et~al.}(2014)\citenamefont {Svore},
  \citenamefont {Hastings},\ and\ \citenamefont {Freedman}}]{SHF14}%
  \BibitemOpen
  \bibfield  {author} {\bibinfo {author} {\bibfnamefont {K.~M.}\ \bibnamefont
  {Svore}}, \bibinfo {author} {\bibfnamefont {M.~B.}\ \bibnamefont {Hastings}},
  \ and\ \bibinfo {author} {\bibfnamefont {M.}~\bibnamefont {Freedman}},\
  }\bibfield  {title} {\enquote {\bibinfo {title} {Faster phase estimation},}\
  }\href@noop {} {\bibfield  {journal} {\bibinfo  {journal} {Quantum
  Information \& Computation}\ }\textbf {\bibinfo {volume} {14}},\ \bibinfo
  {pages} {306} (\bibinfo {year} {2014})}\BibitemShut {NoStop}%
\bibitem [{\citenamefont {Granade}\ \emph {et~al.}(2012)\citenamefont
  {Granade}, \citenamefont {Ferrie}, \citenamefont {Wiebe},\ and\ \citenamefont
  {Cory}}]{granade_robust_2012}%
  \BibitemOpen
  \bibfield  {author} {\bibinfo {author} {\bibfnamefont {C.~E.}\ \bibnamefont
  {Granade}}, \bibinfo {author} {\bibfnamefont {C.}~\bibnamefont {Ferrie}},
  \bibinfo {author} {\bibfnamefont {N.}~\bibnamefont {Wiebe}}, \ and\ \bibinfo
  {author} {\bibfnamefont {D.~G.}\ \bibnamefont {Cory}},\ }\bibfield  {title}
  {\enquote {\bibinfo {title} {Robust online {H}amiltonian learning},}\ }\href
  {\doibase 10.1088/1367-2630/14/10/103013} {\bibfield  {journal} {\bibinfo
  {journal} {New Journal of Physics}\ }\textbf {\bibinfo {volume} {14}},\
  \bibinfo {pages} {103013} (\bibinfo {year} {2012})}\BibitemShut {NoStop}%
\bibitem [{\citenamefont {Ferrie}(2014)}]{ferrie_high_2014}%
  \BibitemOpen
  \bibfield  {author} {\bibinfo {author} {\bibfnamefont {C.}~\bibnamefont
  {Ferrie}},\ }\bibfield  {title} {\enquote {\bibinfo {title} {High posterior
  density ellipsoids of quantum states},}\ }\href {\doibase
  10.1088/1367-2630/16/2/023006} {\bibfield  {journal} {\bibinfo  {journal}
  {New Journal of Physics}\ }\textbf {\bibinfo {volume} {16}},\ \bibinfo
  {pages} {023006} (\bibinfo {year} {2014})}\BibitemShut {NoStop}%
\bibitem [{\citenamefont {Haykin}(2004)}]{haykin2004kalman}%
  \BibitemOpen
  \bibfield  {author} {\bibinfo {author} {\bibfnamefont {S.}~\bibnamefont
  {Haykin}},\ }\href@noop {} {\emph {\bibinfo {title} {Kalman filtering and
  neural networks}}},\ Vol.~\bibinfo {volume} {47}\ (\bibinfo  {publisher}
  {John Wiley \& Sons},\ \bibinfo {year} {2004})\BibitemShut {NoStop}%
\bibitem [{\citenamefont {Smith}\ \emph {et~al.}(2013)\citenamefont {Smith},
  \citenamefont {Doucet}, \citenamefont {de~Freitas},\ and\ \citenamefont
  {Gordon}}]{smith2013sequential}%
  \BibitemOpen
  \bibfield  {author} {\bibinfo {author} {\bibfnamefont {A.}~\bibnamefont
  {Smith}}, \bibinfo {author} {\bibfnamefont {A.}~\bibnamefont {Doucet}},
  \bibinfo {author} {\bibfnamefont {N.}~\bibnamefont {de~Freitas}}, \ and\
  \bibinfo {author} {\bibfnamefont {N.}~\bibnamefont {Gordon}},\ }\href@noop {}
  {\emph {\bibinfo {title} {Sequential Monte Carlo methods in practice}}}\
  (\bibinfo  {publisher} {Springer Science \& Business Media},\ \bibinfo {year}
  {2013})\BibitemShut {NoStop}%
\bibitem [{\citenamefont {Isard}\ and\ \citenamefont
  {Blake}(1998)}]{isard_condensationconditional_1998}%
  \BibitemOpen
  \bibfield  {author} {\bibinfo {author} {\bibfnamefont {M.}~\bibnamefont
  {Isard}}\ and\ \bibinfo {author} {\bibfnamefont {A.}~\bibnamefont {Blake}},\
  }\bibfield  {title} {\enquote {\bibinfo {title} {{CONDENSATION} ---
  {Conditional} {Density} {Propagation} for {Visual} {Tracking}},}\ }\href
  {\doibase 10.1023/A:1008078328650} {\bibfield  {journal} {\bibinfo  {journal}
  {International Journal of Computer Vision}\ }\textbf {\bibinfo {volume}
  {29}},\ \bibinfo {pages} {5} (\bibinfo {year} {1998})}\BibitemShut {NoStop}%
\bibitem [{\citenamefont {Shulman}\ \emph {et~al.}(2014)\citenamefont
  {Shulman}, \citenamefont {Harvey}, \citenamefont {Nichol}, \citenamefont
  {Bartlett}, \citenamefont {Doherty}, \citenamefont {Umansky},\ and\
  \citenamefont {Yacoby}}]{shulman_suppressing_2014}%
  \BibitemOpen
  \bibfield  {author} {\bibinfo {author} {\bibfnamefont {M.~D.}\ \bibnamefont
  {Shulman}}, \bibinfo {author} {\bibfnamefont {S.~P.}\ \bibnamefont {Harvey}},
  \bibinfo {author} {\bibfnamefont {J.~M.}\ \bibnamefont {Nichol}}, \bibinfo
  {author} {\bibfnamefont {S.~D.}\ \bibnamefont {Bartlett}}, \bibinfo {author}
  {\bibfnamefont {A.~C.}\ \bibnamefont {Doherty}}, \bibinfo {author}
  {\bibfnamefont {V.}~\bibnamefont {Umansky}}, \ and\ \bibinfo {author}
  {\bibfnamefont {A.}~\bibnamefont {Yacoby}},\ }\bibfield  {title} {\enquote
  {\bibinfo {title} {Suppressing qubit dephasing using real-time {Hamiltonian}
  estimation},}\ }\href {\doibase 10.1038/ncomms6156} {\bibfield  {journal}
  {\bibinfo  {journal} {Nature Communications}\ }\textbf {\bibinfo {volume}
  {5}},\ \bibinfo {pages} {5156} (\bibinfo {year} {2014})}\BibitemShut
  {NoStop}%
\bibitem [{\citenamefont {Casagrande}(2014)}]{casagrande_design_2014}%
  \BibitemOpen
  \bibfield  {author} {\bibinfo {author} {\bibfnamefont {S.}~\bibnamefont
  {Casagrande}},\ }\href {https://uwspace.uwaterloo.ca/handle/10012/8281}
  {\enquote {\bibinfo {title} {On design and testing of a spectrometer based on
  an {FPGA} development board for use with optimal control theory and high-{Q}
  resonators},}\ } (\bibinfo {year} {2014})\BibitemShut {NoStop}%
\bibitem [{\citenamefont {Hornibrook}\ \emph {et~al.}(2015)\citenamefont
  {Hornibrook}, \citenamefont {Colless}, \citenamefont {Conway~Lamb},
  \citenamefont {Pauka}, \citenamefont {Lu}, \citenamefont {Gossard},
  \citenamefont {Watson}, \citenamefont {Gardner}, \citenamefont {Fallahi},
  \citenamefont {Manfra},\ and\ \citenamefont
  {Reilly}}]{hornibrook_cryogenic_2015}%
  \BibitemOpen
  \bibfield  {author} {\bibinfo {author} {\bibfnamefont {J.}~\bibnamefont
  {Hornibrook}}, \bibinfo {author} {\bibfnamefont {J.}~\bibnamefont {Colless}},
  \bibinfo {author} {\bibfnamefont {I.}~\bibnamefont {Conway~Lamb}}, \bibinfo
  {author} {\bibfnamefont {S.}~\bibnamefont {Pauka}}, \bibinfo {author}
  {\bibfnamefont {H.}~\bibnamefont {Lu}}, \bibinfo {author} {\bibfnamefont
  {A.}~\bibnamefont {Gossard}}, \bibinfo {author} {\bibfnamefont
  {J.}~\bibnamefont {Watson}}, \bibinfo {author} {\bibfnamefont
  {G.}~\bibnamefont {Gardner}}, \bibinfo {author} {\bibfnamefont
  {S.}~\bibnamefont {Fallahi}}, \bibinfo {author} {\bibfnamefont
  {M.}~\bibnamefont {Manfra}}, \ and\ \bibinfo {author} {\bibfnamefont
  {D.}~\bibnamefont {Reilly}},\ }\bibfield  {title} {\enquote {\bibinfo {title}
  {Cryogenic {Control} {Architecture} for {Large}-{Scale} {Quantum}
  {Computing}},}\ }\href {\doibase 10.1103/PhysRevApplied.3.024010} {\bibfield
  {journal} {\bibinfo  {journal} {Physical Review Applied}\ }\textbf {\bibinfo
  {volume} {3}},\ \bibinfo {pages} {024010} (\bibinfo {year}
  {2015})}\BibitemShut {NoStop}%
\bibitem [{\citenamefont {Opper}\ and\ \citenamefont
  {Winther}(1998)}]{opper1998bayesian}%
  \BibitemOpen
  \bibfield  {author} {\bibinfo {author} {\bibfnamefont {M.}~\bibnamefont
  {Opper}}\ and\ \bibinfo {author} {\bibfnamefont {O.}~\bibnamefont
  {Winther}},\ }\bibfield  {title} {\enquote {\bibinfo {title} {A bayesian
  approach to on-line learning},}\ }\href
  {http://dl.acm.org/citation.cfm?id=304710.304756} {\bibfield  {journal}
  {\bibinfo  {journal} {On-line Learning in Neural Networks}\ ,\ \bibinfo
  {pages} {363}} (\bibinfo {year} {1998})}\BibitemShut {NoStop}%
\bibitem [{\citenamefont {Wiebe}\ \emph
  {et~al.}(2014{\natexlab{a}})\citenamefont {Wiebe}, \citenamefont {Granade},
  \citenamefont {Ferrie},\ and\ \citenamefont {Cory}}]{wiebe_hamiltonian_2014}%
  \BibitemOpen
  \bibfield  {author} {\bibinfo {author} {\bibfnamefont {N.}~\bibnamefont
  {Wiebe}}, \bibinfo {author} {\bibfnamefont {C.}~\bibnamefont {Granade}},
  \bibinfo {author} {\bibfnamefont {C.}~\bibnamefont {Ferrie}}, \ and\ \bibinfo
  {author} {\bibfnamefont {D.}~\bibnamefont {Cory}},\ }\bibfield  {title}
  {\enquote {\bibinfo {title} {Hamiltonian learning and certification using
  quantum resources},}\ }\href {\doibase 10.1103/PhysRevLett.112.190501}
  {\bibfield  {journal} {\bibinfo  {journal} {Physical Review Letters}\
  }\textbf {\bibinfo {volume} {112}},\ \bibinfo {pages} {190501} (\bibinfo
  {year} {2014}{\natexlab{a}})}\BibitemShut {NoStop}%
\bibitem [{\citenamefont {Peruzzo}\ \emph {et~al.}(2014)\citenamefont
  {Peruzzo}, \citenamefont {McClean}, \citenamefont {Shadbolt}, \citenamefont
  {Yung}, \citenamefont {Zhou}, \citenamefont {Love}, \citenamefont
  {Aspuru-Guzik},\ and\ \citenamefont {O’Brien}}]{PMS+14}%
  \BibitemOpen
  \bibfield  {author} {\bibinfo {author} {\bibfnamefont {A.}~\bibnamefont
  {Peruzzo}}, \bibinfo {author} {\bibfnamefont {J.}~\bibnamefont {McClean}},
  \bibinfo {author} {\bibfnamefont {P.}~\bibnamefont {Shadbolt}}, \bibinfo
  {author} {\bibfnamefont {M.-H.}\ \bibnamefont {Yung}}, \bibinfo {author}
  {\bibfnamefont {X.-Q.}\ \bibnamefont {Zhou}}, \bibinfo {author}
  {\bibfnamefont {P.~J.}\ \bibnamefont {Love}}, \bibinfo {author}
  {\bibfnamefont {A.}~\bibnamefont {Aspuru-Guzik}}, \ and\ \bibinfo {author}
  {\bibfnamefont {J.~L.}\ \bibnamefont {O’Brien}},\ }\bibfield  {title}
  {\enquote {\bibinfo {title} {A variational eigenvalue solver on a photonic
  quantum processor},}\ }\href {\doibase 10.1038/ncomms5213} {\bibfield
  {journal} {\bibinfo  {journal} {Nature Communications}\ }\textbf {\bibinfo
  {volume} {5}},\ \bibinfo {pages} {4213} (\bibinfo {year} {2014})}\BibitemShut
  {NoStop}%
\bibitem [{\citenamefont {McClean}\ \emph {et~al.}(2014)\citenamefont
  {McClean}, \citenamefont {Babbush}, \citenamefont {Love},\ and\ \citenamefont
  {Aspuru-Guzik}}]{MBL+14}%
  \BibitemOpen
  \bibfield  {author} {\bibinfo {author} {\bibfnamefont {J.~R.}\ \bibnamefont
  {McClean}}, \bibinfo {author} {\bibfnamefont {R.}~\bibnamefont {Babbush}},
  \bibinfo {author} {\bibfnamefont {P.~J.}\ \bibnamefont {Love}}, \ and\
  \bibinfo {author} {\bibfnamefont {A.}~\bibnamefont {Aspuru-Guzik}},\
  }\bibfield  {title} {\enquote {\bibinfo {title} {Exploiting locality in
  quantum computation for quantum chemistry},}\ }\href {\doibase
  10.1021/jz501649m} {\bibfield  {journal} {\bibinfo  {journal} {Journal of
  Physical Chemistry Letters}\ }\textbf {\bibinfo {volume} {5}},\ \bibinfo
  {pages} {4368} (\bibinfo {year} {2014})}\BibitemShut {NoStop}%
\bibitem [{\citenamefont {Wecker}\ \emph {et~al.}(2015)\citenamefont {Wecker},
  \citenamefont {Hastings},\ and\ \citenamefont {Troyer}}]{WHT15}%
  \BibitemOpen
  \bibfield  {author} {\bibinfo {author} {\bibfnamefont {D.}~\bibnamefont
  {Wecker}}, \bibinfo {author} {\bibfnamefont {M.}~\bibnamefont {Hastings}}, \
  and\ \bibinfo {author} {\bibfnamefont {M.}~\bibnamefont {Troyer}},\
  }\bibfield  {title} {\enquote {\bibinfo {title} {Towards practical quantum
  variational algorithms},}\ }\href@noop {} {\bibfield  {journal} {\bibinfo
  {journal} {arXiv:1507.08969v1}\ } (\bibinfo {year} {2015})}\BibitemShut
  {NoStop}%
\bibitem [{\citenamefont {Ferrie}\ \emph {et~al.}(2013)\citenamefont {Ferrie},
  \citenamefont {Granade},\ and\ \citenamefont {Cory}}]{ferrie_how_2013}%
  \BibitemOpen
  \bibfield  {author} {\bibinfo {author} {\bibfnamefont {C.}~\bibnamefont
  {Ferrie}}, \bibinfo {author} {\bibfnamefont {C.~E.}\ \bibnamefont {Granade}},
  \ and\ \bibinfo {author} {\bibfnamefont {D.~G.}\ \bibnamefont {Cory}},\
  }\bibfield  {title} {\enquote {\bibinfo {title} {How to best sample a
  periodic probability distribution, or on the accuracy of {H}amiltonian
  finding strategies},}\ }\href
  {http://link.springer.com/article/10.1007/s11128-012-0407-6} {\bibfield
  {journal} {\bibinfo  {journal} {Quantum Information Processing}\ }\textbf
  {\bibinfo {volume} {12}},\ \bibinfo {pages} {611} (\bibinfo {year}
  {2013})}\BibitemShut {NoStop}%
\bibitem [{\citenamefont {Wiebe}\ \emph {et~al.}(2015)\citenamefont {Wiebe},
  \citenamefont {Granade},\ and\ \citenamefont {Cory}}]{WGC15}%
  \BibitemOpen
  \bibfield  {author} {\bibinfo {author} {\bibfnamefont {N.}~\bibnamefont
  {Wiebe}}, \bibinfo {author} {\bibfnamefont {C.}~\bibnamefont {Granade}}, \
  and\ \bibinfo {author} {\bibfnamefont {D.~G.}\ \bibnamefont {Cory}},\
  }\bibfield  {title} {\enquote {\bibinfo {title} {Quantum bootstrapping via
  compressed quantum hamiltonian learning},}\ }\href {\doibase
  10.1088/1367-2630/17/2/022005} {\bibfield  {journal} {\bibinfo  {journal}
  {New Journal of Physics}\ }\textbf {\bibinfo {volume} {17}},\ \bibinfo
  {pages} {022005} (\bibinfo {year} {2015})}\BibitemShut {NoStop}%
\bibitem [{\citenamefont {Wiebe}\ \emph
  {et~al.}(2014{\natexlab{b}})\citenamefont {Wiebe}, \citenamefont {Granade},
  \citenamefont {Ferrie},\ and\ \citenamefont {Cory}}]{wiebe_quantum_2014-1}%
  \BibitemOpen
  \bibfield  {author} {\bibinfo {author} {\bibfnamefont {N.}~\bibnamefont
  {Wiebe}}, \bibinfo {author} {\bibfnamefont {C.}~\bibnamefont {Granade}},
  \bibinfo {author} {\bibfnamefont {C.}~\bibnamefont {Ferrie}}, \ and\ \bibinfo
  {author} {\bibfnamefont {D.}~\bibnamefont {Cory}},\ }\bibfield  {title}
  {\enquote {\bibinfo {title} {Quantum {H}amiltonian learning using imperfect
  quantum resources},}\ }\href {\doibase 10.1103/PhysRevA.89.042314} {\bibfield
   {journal} {\bibinfo  {journal} {Physical Review A}\ }\textbf {\bibinfo
  {volume} {89}},\ \bibinfo {pages} {042314} (\bibinfo {year}
  {2014}{\natexlab{b}})}\BibitemShut {NoStop}%
\bibitem [{\citenamefont {Goodman}(2008)}]{goodman_dirty_2008}%
  \BibitemOpen
  \bibfield  {author} {\bibinfo {author} {\bibfnamefont {S.}~\bibnamefont
  {Goodman}},\ }\bibfield  {title} {\enquote {\bibinfo {title} {A dirty dozen:
  {Twelve} {p}-value misconceptions},}\ }\href {\doibase
  10.1053/j.seminhematol.2008.04.003} {\bibfield  {journal} {\bibinfo
  {journal} {Seminars in Hematology}\ }\bibinfo {series} {Interpretation of
  {Quantitative} {Research}},\ \textbf {\bibinfo {volume} {45}},\ \bibinfo
  {pages} {135} (\bibinfo {year} {2008})}\BibitemShut {NoStop}%
\bibitem [{\citenamefont {Hoekstra}\ \emph {et~al.}(2014)\citenamefont
  {Hoekstra}, \citenamefont {Morey}, \citenamefont {Rouder},\ and\
  \citenamefont {Wagenmakers}}]{hoekstra_robust_2014}%
  \BibitemOpen
  \bibfield  {author} {\bibinfo {author} {\bibfnamefont {R.}~\bibnamefont
  {Hoekstra}}, \bibinfo {author} {\bibfnamefont {R.~D.}\ \bibnamefont {Morey}},
  \bibinfo {author} {\bibfnamefont {J.~N.}\ \bibnamefont {Rouder}}, \ and\
  \bibinfo {author} {\bibfnamefont {E.-J.}\ \bibnamefont {Wagenmakers}},\
  }\bibfield  {title} {\enquote {\bibinfo {title} {Robust misinterpretation of
  confidence intervals},}\ }\href {\doibase 10.3758/s13423-013-0572-3}
  {\bibfield  {journal} {\bibinfo  {journal} {Psychonomic Bulletin \& Review}\
  ,\ \bibinfo {pages} {1}} (\bibinfo {year} {2014})}\BibitemShut {NoStop}%
\end{thebibliography}
\end{document}